\documentclass[journal]{IEEEtran}
\IEEEoverridecommandlockouts
\usepackage{amsmath,amsfonts,amssymb,euscript, graphicx, 
epsfig,
enumerate,float,afterpage, subfigure, ifthen}%

\newtheorem{thm}{Theorem}

\newtheorem{lem}{Lemma}

\newcommand{\expect}[1]{\mathbb{E}\left\{#1\right\}}
\newcommand{\defequiv}{\mbox{\raisebox{-.3ex}{$\overset{\vartriangle}{=}$}}}

\newcommand{\bv}[1]{{\boldsymbol{#1} }}
\newcommand{\script}[1]{{{\cal{#1} }}}

\begin{document}

\title
  {Stock Market Trading via Stochastic Network Optimization}
\author{Michael J. Neely \\ University of Southern California\\ http://www-rcf.usc.edu/$\sim$mjneely\\
\thanks{This material is supported in part  by one or more of 
the following: the DARPA IT-MANET program
grant W911NF-07-0028, 
the NSF Career grant CCF-0747525.}}
\markboth{}{Neely}

\maketitle

\begin{abstract} 
We consider the problem of dynamic buying and selling 
of shares from a collection of $N$ stocks with random price fluctuations. 
To limit investment risk, we place an upper bound on the total number
of shares kept at any time.   Assuming that prices evolve according
to an ergodic process with a mild \emph{decaying memory property}, and assuming 
constraints on the total number of shares that can be bought and sold at
any time, we develop 
a trading policy that comes arbitrarily close to achieving the profit of an ideal policy that has
perfect knowledge of future events.   Proximity to the optimal profit comes with a corresponding
tradeoff in the maximum required stock level and in the timescales associated with convergence.  
We then consider arbitrary (possibly non-ergodic) price processes, and show that the same 
algorithm comes close to the profit of a frame based policy that can look 
a fixed number of slots into the future. 
Our analysis uses 
techniques of Lyapunov Optimization that we originally 
developed for stochastic network optimization problems.
\end{abstract} 
\begin{keywords} 
Queueing analysis, stochastic control, universal algorithms
\end{keywords}

\section{Introduction}

This paper considers the problem of stock trading in an economic market with $N$ stocks.
We treat the problem in discrete time with normalized time slots $t \in \{0, 1, 2, \ldots\}$, where
buying and selling transactions are conducted on each slot.  
Let $\bv{Q}(t) = (Q_1(t), \ldots, Q_N(t))$ 
be a vector of the current number of shares owned of each stock, called the \emph{stock queue}. That is, 
for each $n \in \{1, \ldots, N\}$, the value of $Q_n(t)$ is an integer that represents the number of 
shares of stock $n$.   
Stock prices are given by a vector $\bv{p}(t) = (p_1(t), \ldots, p_N(t))$ and are assumed
to evolve randomly,  with mild assumptions to be made precise in later  sections. 
Each buy and sell transaction incurs trading costs.  Stocks can be sold and purchased on every slot.
Let  $\phi(t)$ represent the net profit on slot $t$ (after transaction costs are paid). 
The goal is to design a trading policy that maximizes
the long term time average of $\phi(t)$.    

For this system model, we enforce the additional constraint that 
at most $\mu_n^{max}$ shares of each stock $n$ can be bought and sold on a given slot.  
This ensures that our trading decisions only gradually change the portfolio allocation. 
While this $\mu_n^{max}$ constraint can significantly limit the ability to take advantage of 
desirable prices, and hence limits 
the maximum possible long term profit, we show that it can also reduce investment
risk. 
Specifically, subject to the $\mu_n^{max}$ constraint, 
we develop an algorithm that achieves a time average profit that 
is arbitrarily close to optimal, with a tradeoff in the maximum number of shares $Q_n^{max}$ 
required for stock $n$.    The $Q_n^{max}$ values can be chosen as desired to limit  
the  losses  from a  potential collapse of  one or more of the stocks.  
 It also impacts the timescales over which profit is accumulated, where smaller $Q_n^{max}$ levels
 lead to faster convergence times.     
 
It is important to note that long term wealth typically grows \emph{exponentially} when the $Q_n^{max}$ and 
$\mu_n^{max}$ constraints are removed.  In contrast, it can be shown that 
these $Q_n^{max}$ and $\mu_n^{max}$ constraints restrict wealth to at most a  \emph{linear growth}. 
Therefore, using $Q_n^{max}$ and $\mu_n^{max}$ to limit investment risk  unfortunately has a dramatic 
impact  on the long term growth curve.  
However, our ability to bound the timescales over which wealth is earned suggests that our strategy 
may be useful in cases when, in addition to a good long-term return, we also desire noticeable and consistent
short-term gains. 
At the end of this paper, we briefly describe a  
modified strategy that increases $Q_n^{max}$ and $\mu_n^{max}$ as wealth progresses, with the goal of 
achieving noticeable short-term gains while enabling exponential wealth increase.

Our approach uses the Lyapunov optimization theory developed
for stochastic queueing networks in our previous work 
\cite{now}\cite{neely-energy-it}\cite{neely-thesis}.   
Specifically, the work \cite{now}\cite{neely-energy-it}\cite{neely-thesis} develops
resource allocation and scheduling policies for communication and queueing networks with 
random traffic and channels.  The policies can maximize time average throughput-utility  
and minimize time average power expenditure, as well as optimize more general time average attributes, 
without a-priori knowledge of the traffic and channel probabilities.  The algorithms
continuously adapt to emerging conditions, and are robust to non-ergodic changes in the probability
distributions \cite{neely-mesh}.  This suggests that similar control techniques can be used successfully
for stock trading problems.  The difference is that the queues associated with stock shares are controlled
to have \emph{positive drift} (pushing them towards the maximum queue size), rather than \emph{negative
drift} (which would push them in the direction of the empty state).  

The \emph{Dynamic Trading Algorithm} that we develop from these techniques can be intuitively
viewed as a variation on a theme of \emph{dollar cost averaging}, where price downturns are exploited
by purchasing more stock.   However, the actual amount of stock that we buy and sell on each slot 
is determined by a constrained optimization of a max-weight functional 
that incorporates transaction costs, current prices, and current stock queue levels.

Much prior work on financial analysis and portfolio optimization assumes a known probability model for 
stock prices.   Classical portfolio optimization techniques by Markowitz \cite{markowitz52} and Sharpe \cite{sharpe63} 
construct portfolio allocations over $N$ stocks to maximize profit subject to variance constraints (which model risk)
over one investment period (see also \cite{risk-book} and references therein). 
Solutions to this problem can be calculated if the mean and covariance of stock returns
are known.    Samuelson considers multi-period problems in \cite{samuelson69} using 
dynamic programming, assuming a known product form distribution for investment returns. 
 Cover in \cite{cover-stock84}
develops an iterative
procedure that converges to the constant portfolio allocation that maximizes the expected log investment return, 
assuming a known probability distribution that is the same on each period.  Recent work by
Rudoy and Rohrs in \cite{rudoy-cdc} \cite{rudoy-thesis} 
considers risk-aware optimization  with a more complex \emph{cointegrated vector
autorgressive} assumption on stock processes,  and uses Monte Carlo
simulations over historical stock trajectories to inform stochastic decisions. 
Stochastic models of stock prices using L\'{e}vy processes and  multi-fractal processes
are considered in \cite{risk-book} \cite{stock-multi-fractal} \cite{distribution-differences} and references therein. 

 A significant departure from this work is the 
 \emph{universal stock trading paradigm},  as exemplified in 
 prior works of Cover and Gluss \cite{cover-gluss86},  Larson \cite{larson-stock-thesis}, 
 Cover \cite{cover-universal}, Merhav and Feder \cite{merhav-universal}, and 
 Cover and Ordentlich \cite{universal-stock2} \cite{max-min-stock}, 
 where trading algorithms are developed and shown to provide
analytical guarantees for \emph{any sample path} of stock prices.  Specifically, the work in 
\cite{cover-gluss86}-\cite{max-min-stock} seeks to find a non-anticipating trading algorithm that yields
the same growth exponent as the best constant portfolio allocation, where the constant can be optimized with full 
knowledge of the future.   The works in \cite{cover-gluss86}\cite{larson-stock-thesis} develop algorithms 
that come close to the optimal exponent, and the work in \cite{cover-universal} \emph{achieves} the optimal exponent
under a mild \emph{active stock} assumption on the price sample paths.  Similar results are derived
in \cite{merhav-universal} using a general framework of sequential decision theory. 
Related results are derived in \cite{universal-stock2} \cite{max-min-stock}
without the \emph{active stock} assumption, 
where \cite{max-min-stock} also treats
max-min performance when stock prices are chosen by an adversary.

Our work is similar in spirit to this universal trading paradigm, in 
that we do not base decisions on a known (or estimated) probability distribution.   However, our context 
and solution methodology is very different.  Indeed, 
the works in \cite{cover-gluss86}-\cite{max-min-stock} assume that the entire stock portfolio can be sold and reallocated
on every time period, and 
allow stock holdings to grow arbitrarily large.  This means that the accumulated profit is always 
at risk of one or more stock failures.  In our work, we take a more conservative approach that restricts reallocation to
gradual changes, and that pockets profits while holding 
no more than $Q_n^{max}$ shares of each stock $n$.   
We also explicitly account for trading costs and integer constraints on stock shares, which is
not considered in the works \cite{cover-gluss86}-\cite{max-min-stock}.  
In this context, we first design an algorithm under the assumption that prices are ergodic with an
unknown distribution.   In this case, we develop a simple non-anticipating 
algorithm that comes arbitrarily close to the optimal time
average profit that could be earned by an ideal policy with complete knowledge of the future. 
The ideal policy used for comparison 
can make different allocations at different times, and is not restricted to constant allocations 
as considered in \cite{cover-gluss86}-\cite{max-min-stock}. 
We then show that the \emph{same} algorithm can be used for general price sample paths, 
even non-ergodic sample paths without well defined time averages.  A more conservative 
guarantee is shown in this case:  The  algorithm 
yields profit that is arbitrarily close to that of a frame based policy with ``$T$-slot lookahead,'' where the future
is known up to $T$ slots.   Our approach
is inspired by Lyapunov optimization and 
decision theory for stochastic queueing networks \cite{now}.  However, the 
Lyapunov theory we use here involves sample path techniques that 
are different from those in \cite{now}.  These techniques might have broader impacts on  
queueing problems in other areas.

In the next section we present the system model. 
In Section \ref{section:iid} we develop the Dynamic Trading Algorithm and analyze performance
for the simple (and possibly unrealistic) case when price vectors $\bv{p}(t)$ are ergodic and i.i.d. over slots. 
While this i.i.d. case does \emph{not} accurately model actual stock prices, its analysis provides
valuable insight.  Section \ref{section:non-iid} expands the analysis to show the same algorithm can handle 
more general ergodic processes with a mild
\emph{decaying memory property}.  
Section \ref{section:arbitrary-prices} shows the algorithm
also provides  performance guarantees for completely arbitrary price processes (possibly non-ergodic). 
A simple enhancement that reduces startup cost is treated in Section \ref{section:place-holder}, and
Section \ref{section:splits}  briefly considers an extension that allows for exponential wealth increase
by gradually scaling the $\mu_n^{max}$ and $Q_n^{max}$ parameters. 

\section{System Model} 

Let $\bv{A}(t) = (A_1(t), \ldots, A_N(t))$
be a vector of decision variables representing the number of new shares purchased for each stock on slot $t$, 
and let $\bv{\mu}(t) = (\mu_1(t), \ldots, \mu_N(t))$ be a vector representing 
the number of shares sold on slot $t$.  The values
$A_n(t)$ and $\mu_n(t)$ are non-negative integers for each $n \in \{1, \ldots, N\}$. 
Each purchase of $A$ new shares of stock $n$ incurs a transaction cost $b_n(A)$ (called the \emph{buying cost function}).  Likewise, each sale of $\mu$ shares of stock $n$ incurs a transaction cost $s_n(\mu)$
(called the \emph{selling cost function}).  The functions $b_n(A)$ and $s_n(\mu)$ are arbitrary, and are 
assumed only to satisfy $b_n(0) = s_n(0) = 0$,  and to be non-negative, non-decreasing,  
and bounded by finite constants
$b_n^{max}$ and $s_n^{max}$,  
so that: 
\begin{eqnarray*}
0 \leq b_n(A) \leq b_n^{max} & \mbox{ for $0 \leq A \leq \mu_n^{max}$} \\
0 \leq s_n(\mu) \leq s_n^{max} & \mbox{ for $0 \leq \mu \leq \mu_n^{max}$} 
\end{eqnarray*}
where for each $n \in \{1, \ldots, N\}$, $\mu_n^{max}$ is a positive integer that limits the amount of shares of stock $n$
that can be bought and sold on slot $t$.  

\subsection{Example Transaction Cost Functions} 

The functions $b_n(A)$ might be \emph{linear}, representing a transaction fee that 
charges per share purchased.  Another example is a \emph{fixed cost model} with some 
fixed positive fee $b_n$, so that: 
\[ b_n(A) = \left\{ \begin{array}{ll}
                          b_n &\mbox{ if $A > 0$} \\
                             0  & \mbox{ if $A=0$} 
                            \end{array}
                                 \right. \]
 Similar models can be used for the $s_n(\mu)$ function. 
 The simplest model of all is the \emph{zero transaction cost model} where the functions $b_n(A)$ and 
$s_n(\mu)$ are identically zero. 

\subsection{System Dynamics} 

The stock price vector $\bv{p}(t)$ is assumed to be a
random vector process that takes values in some finite set $\script{P} \subset \mathbb{R}^N$, where $\script{P}$
can have an arbitrarily large number of elements.\footnote{The cardinality of the set $\script{P}$ does not enter into our
analysis.  We assume it is finite only for the convenience of claiming that the supremum time average profit 
$\phi^{opt}$ is achievable by a single ``$p$-only'' policy, as described in Section \ref{section:p-only}.  Theorems
\ref{thm:1}, \ref{thm:2}, \ref{thm:3} are unchanged if the set  $\script{P}$  is infinite, although the proofs of Theorems \ref{thm:1} 
and \ref{thm:2} would require an additional limiting argument over $p$-only policies that approach $\phi^{opt}$.}
For each $n$, let  $p_n^{max}$ represent a bound on  $p_n(t)$,   so that: 
\begin{eqnarray}
 0 \leq p_n(t) \leq p_n^{max} &  \mbox{ for all $t$ and all $n\in\{1, \ldots, N\}$} \label{eq:price-bound}
 \end{eqnarray}
We assume that 
buying and selling decisions can be made on each slot $t$ based on knowledge of $\bv{p}(t)$. 
The selling decision variables $\bv{\mu}(t)$ are made every slot $t$ subject to the following constraints: 
\begin{eqnarray} 
\mu_n(t) \in \{0, 1, \ldots, \mu_n^{max} \} & \mbox{ for all $n \in \{1, \ldots, N\}$} \label{eq:mu1} \\
\mu_n(t)p_n(t) \geq s_n(\mu_n(t)) & \mbox{ for all $n \in \{1, \ldots, N\}$} \label{eq:mu2} \\
\mu_n(t)\leq Q_n(t)   & \mbox{ for all $n \in \{1, \ldots, N\}$}  \label{eq:mu3} 
\end{eqnarray}
Constraint (\ref{eq:mu1}) ensures that no more than $\mu_n^{max}$ shares can be sold of any stock on a single slot. 
Constraint (\ref{eq:mu2}) restricts to the reasonable case when the money earned
from the sale of a stock must be larger than the transaction fee associated with the sale (violating
this constraint would clearly be sub-optimal).\footnote{Constraint (\ref{eq:mu2}) can be augmented
by allowing equality only if $\mu_n(t) =0$.}  Constraint (\ref{eq:mu3}) 
requires the number of shares sold to be less than or equal to the current number owned.

The buying decision variables $\bv{A}(t)$ are constrained as follows: 
\begin{eqnarray} 
A_n(t) \in \{0, 1, 2, \ldots, \mu_n^{max} \} \: \: \mbox{ for all $n \in \{1, \ldots, N\}$} \label{eq:A1} \\
&\hspace{-3in} \sum_{n=1}^N A_n(t)p_n(t) \leq x & \label{eq:A2} 
\end{eqnarray} 
where $x$ is a positive value that bounds the total amount of money used for purchases on slot $t$. 
For simplicity, we assume there is always at least a minimum of $x$ and 
$\sum_{n=1}^N [\mu_n^{max}p_n^{max} + b_n(\mu_n^{max})]$
dollars available for making purchasing decisions. This model can be augmented by 
adding a \emph{checking account queue} $Q_0(t)$ from which we must draw money to make purchases, although
we omit this aspect for brevity. 

The resulting queueing dynamics for the stock queues $Q_n(t)$ for $n \in \{1, \ldots, N\}$
are thus: 
\begin{eqnarray}
Q_n(t+1) =  \max[Q_n(t) - \mu_n(t) + A_n(t), 0]  \label{eq:q-dynamics} 
\end{eqnarray}
Strictly speaking, the $\max[\cdot, 0]$ operator in the above 
dynamic equation is redundant,  because the constraint  (\ref{eq:mu3}) 
ensures that the  argument inside the $\max[\cdot, 0]$ operator is non-negative. 
However, the $\max[\cdot, 0]$ shall be useful for mathematical analysis when we 
compare our strategy to that of a queue-independent strategy that neglects constraint (\ref{eq:mu3}). 

\subsection{The Maximum Profit Objective} 

Define $\phi(t)$ as the net profit on slot $t$: 
\begin{eqnarray}
\phi(t) &\defequiv& \sum_{n=1}^{N}[\mu_n(t)p_n(t) - s_n(\mu_n(t))]  \nonumber \\
&& -\sum_{n=1}^{N} [A_n(t)p_n(t) + b_n(A_n(t))]  \label{eq:phi} 
\end{eqnarray} 
Define $\overline{\phi}$ as the time average expected value of $\phi(t)$ under a given trading
algorithm (temporarily assumed to have a well defined limit): 
\[ \overline{\phi} \defequiv \lim_{t\rightarrow\infty} \frac{1}{t}\sum_{\tau=0}^{t-1} \expect{\phi(\tau)} \]
The goal is to design a trading policy that maximizes $\overline{\phi}$.  
It is clear that the trivial  strategy that chooses $\bv{\mu}(t) = \bv{A}(t) = \bv{0}$ for all $t$
yields $\phi(t) = 0$ for all $t$, and results in  $\overline{\phi} = 0$. 
Therefore, we desire our algorithm to produce a long term profit that satisfies
$\overline{\phi} > 0$.

\subsection{Discussion of Constraints} 
If we set $x \defequiv \sum_{n=1}^N [\mu_n^{max}p_n^{max} + b_n(\mu_n^{max})]$, then 
constraint (\ref{eq:A2}) is redundant and
can be removed. In this case, the multi-stock problem 
completely decouples into separate problems of optimally trading on each of the individual stocks. 
Trading on just a single stock is itself an important problem that can be viewed as a special 
case of our system model.    We add the constraint (\ref{eq:A2}) for multi-stock problems
as it can be used to limit 
the total amount spent on new purchases on a single slot.  The constraint (\ref{eq:A2}) 
can lead to a complex decision on each slot that is related to the \emph{bounded knapsack problem}, 
as discussed in Section \ref{section:alg} after the description
of the Dynamic Trading Algorithm.  The formulation can be modified by replacing the constraint (\ref{eq:A2})
with the  following constraint that often yields a simpler implementation: 
\begin{eqnarray} 
&  \sum_{n=1}^N A_n(t) \leq A_{tot} & \label{eq:new-constraint-x}
\end{eqnarray}
where $A_{tot}$ is an integer that bounds the total number of stocks that can be bought on a single slot.

\subsection{The Stochastic Price Vector and $p$-only Policies} \label{section:p-only} 

We first assume the stochastic process $\bv{p}(t)$ has well defined time averages (this is generalized
to non-ergodic models in Section \ref{section:arbitrary-prices}).  Specifically, 
for each price vector $\bv{p}$ in the finite set $\script{P}$, 
we define $\pi(\bv{p})$ as the time average fraction of time that 
$\bv{p}(t) = \bv{p}$, so that: 
\begin{equation} \label{eq:time-avg-p} 
\lim_{t\rightarrow\infty}\frac{1}{t}\sum_{\tau=0}^{t-1} 1\{\bv{p}(\tau) = \bv{p}\} = \pi(\bv{p}) \: \: \mbox{with probability 1} 
\end{equation} 
where $1\{\bv{p}(\tau) = \bv{p}\}$ is an indicator function that is $1$ if $\bv{p}(\tau) = \bv{p}$, and zero otherwise.

 Define a \emph{$p$-only policy} as a buying and selling strategy that chooses
\emph{virtual} decision vectors  $\bv{A}^*(t)$ and $\bv{\mu}^*(t)$
as a stationary and possibly randomized function of $\bv{p}(t)$, constrained only by 
(\ref{eq:mu1})-(\ref{eq:mu2}) and (\ref{eq:A1})-(\ref{eq:A2}). 
That is, the virtual decision vectors $\bv{A}^*(t)$ and $\bv{\mu}^*(t)$ associated with a $p$-only 
policy do not necessarily satisfy the constraint  (\ref{eq:mu3}) 
that is required of the \emph{actual} decision vectors, and hence these decisions can be made independently
of the current stock queue levels. 

Under a given $p$-only policy, define the following time average expectations 
$d_n^*$ and $\phi^*$: 
\begin{eqnarray} 
d_n^* \defequiv \lim_{t\rightarrow\infty} \frac{1}{t}\sum_{\tau=0}^{t-1} \expect{A_n^*(\tau) - \mu_n^*(\tau)}  \label{eq:dn-star} 
\end{eqnarray} 
\begin{eqnarray}
\phi^* \defequiv \lim_{t\rightarrow\infty} \frac{1}{t}\sum_{\tau=0}^{t-1} \mathbb{E}\left\{  \sum_{n=1}^N [\mu_n^*(\tau)p_n(\tau) - s_n(\mu_n^*(\tau))] \right. \hspace{-.2in} \nonumber \\
- \left. \sum_{n=1}^N[ A_n^*(\tau)p_n(\tau) + b_n(A_n^*(\tau))] \right\} \hspace{.2in}  \label{eq:phi-star} 
\end{eqnarray} 
It is easy to see by (\ref{eq:time-avg-p}) that these time averages are well defined for any $p$-only policy. 
For each $n$, the value 
$d_n^*$ represents the \emph{virtual drift} of stock queue $Q_n(t)$ associated with the virtual
decisions $\bv{A}^*(t)$ and $\bv{\mu}^*(t)$. The value $\phi^*$ represents the \emph{virtual profit} 
under virtual decisions $\bv{A}^*(t)$ and $\bv{\mu}^*(t)$.  Note that the trivial $p$-only
policy $\bv{A}^*(t) = \bv{\mu}^*(t) = \bv{0}$ yields $d_n^*= 0$ for all $n$, and $\phi^* = 0$.  
Thus, we can define $\phi^{opt}$ as the 
supremum value of $\phi^*$ over all $p$-only policies that yield $d_n^* \geq 0$ for all $n$, and we note
that $\phi^{opt} \geq 0$.   
Using an argument similar to that given in \cite{neely-energy-it}, it can be shown that: 

\begin{enumerate} 
\item $\phi^{opt}$ is \emph{achievable} by a single $p$-only policy that satisfies $d_n^* = 0$ for all $n\in\{1, \ldots, N\}$. 
\item  $\phi^{opt}$ is greater than or equal to the supremum of the $\limsup$ time average
expectation of $\phi(t)$ that can be achieved over the class of 
all \emph{actual} policies that satisfy the constraints (\ref{eq:mu1})-(\ref{eq:A2}), 
including ideal policies that use perfect information about the future.   Thus, no policy can do better than $\phi^{opt}$.
\end{enumerate} 

That $\phi^{opt}$ is achievable by a single $p$-only policy (rather than by a limit of an infinite sequence of policies)
can be shown using the assumption that the set $\script{P}$ of all price vectors is finite. That $\phi^{opt}$ bounds
the time average profit of \emph{all} policies, including those that have perfect knowledge of the future, can be intuitively 
understood
by noting that the optimal profit is determined only by the time averages $\pi(\bv{p})$.  These time averages are 
the same (with probability 1) regardless of whether or not we know the future. 
The detailed proofs of these results are similar to those in  \cite{neely-energy-it}
and are provided in Appendix C.   In the next section we develop a \emph{Dynamic Trading Algorithm} that satisfies
the constraints (\ref{eq:mu1})-(\ref{eq:A2}) and that does not know the future or the  distribution $\pi(\bv{p})$, yet
yields time average profit that is arbitrarily close to $\phi^{opt}$.  

To develop our Dynamic Trading Algorithm, we first focus on the simple case when the vector $\bv{p}(t)$ is independent 
and identically distributed (i.i.d.) over slots, 
with a general probability distribution $\pi(\bv{p})$. 
This is an overly simplified model and does \emph{not} reflect actual stock time series data.  Indeed, a more accurate
model would be to assume the differences in the logarithm of prices are i.i.d. (see \cite{risk-book} and references therein). 
 However, we show in Section \ref{section:non-iid} that the \emph{same} algorithm developed for the simplified i.i.d. case 
can \emph{also} be used for a general class of ergodic but 
non-i.i.d. processes that have a mild
\emph{decaying memory property} (a property held by all processes that are modulated by finite state Markov chains). 
Section \ref{section:arbitrary-prices} shows the algorithm can also  treat arbitrary (possibly non-ergodic) price 
models.

\subsection{The i.i.d. Model} 

Suppose $\bv{p}(t)$ is i.i.d. over slots with $Pr[\bv{p}(t) = \bv{p}] = \pi(\bv{p})$ for all $\bv{p} \in \script{P}$.  
Because the value $\phi^{opt}$ is achievable by a single $p$-only policy, and because the expected values 
of any $p$-only policy are the same every slot under the i.i.d. model, we have the following: 
There is a $p$-only policy $\bv{A}^*(t)$, $\bv{\mu}^*(t)$ that yields for all $t$ and all $\bv{Q}(t)$: 
\begin{eqnarray} 
\expect{A_n^*(t) - \mu_n^*(t)\left|\right.\bv{Q}(t)} = 0 \label{eq:zero-drift-iid} 
\end{eqnarray} 
and
\begin{eqnarray}
&\hspace{-.7in}\mathbb{E}\left\{\sum_{n=1}^N [\mu_n^*(t)p_n(t) - s_n(\mu_n^*(t))]  \right. \nonumber \\
&\hspace{-.2in}\left.- \sum_{n=1}^N[ A_n^*(t)p_n(t) + b_n(A_n^*(t))]  \left|\right.\bv{Q}(t)\right\} = \phi^{opt} 
\label{eq:opt-profit-iid} 
\end{eqnarray}

\section{Constructing a Dynamic Trading Algorithm} \label{section:iid}

 The goal is to ensure that all stock queues $Q_n(t)$ are maintained at reasonably high levels so that there
 are typically enough shares available to sell if an opportune price should arise.   
 To this end, define
 $\theta_1, \ldots, \theta_n$ as positive real numbers that represent target queue sizes for the stock queues 
 (soon to be 
 related to the maximum queue size).  The particular values $\theta_1, \ldots, \theta_n$ shall be chosen later. 
 As a scalar measure of the distance each queue is away from its target value, we define the following
 \emph{Lyapunov function} $L(\bv{Q}(t))$: 
 \begin{equation} \label{eq:lyap-function} 
  L(\bv{Q}(t)) \defequiv \frac{1}{2}\sum_{n=1}^N(Q_n(t) -\theta_n)^2 
  \end{equation} 
  Suppose that $\bv{Q}(t)$ evolves according to some probability law, and 
  define $\Delta(\bv{Q}(t))$ as the \emph{one-slot conditional Lyapunov drift}:\footnote{Strictly speaking, 
 proper notation is $\Delta(\bv{Q}(t), t)$, as the drift may arise from a non-stationary algorithm. However, 
 we use the simpler notation
 $\Delta(\bv{Q}(t))$ as a formal representation of the right hand side of (\ref{eq:one-slot-drift}).} 
 \begin{equation} \label{eq:one-slot-drift} 
  \Delta(\bv{Q}(t)) \defequiv \expect{L(\bv{Q}(t+1)) - L(\bv{Q}(t))\left|\right.\bv{Q}(t)} 
  \end{equation} 
  As in the stochastic network optimization problems of \cite{now}\cite{neely-energy-it}\cite{neely-thesis}, 
  our approach is to take control actions on each slot $t$ to minimize a bound on the ``drift-minus-reward'' 
  expression: 
  \[ \Delta(\bv{Q}(t)) - V\expect{\phi(t)\left|\right.\bv{Q}(t)} \]
  where $V$ is a positive parameter to be chosen as desired to affect the proximity to the optimal 
  time average profit $\phi^{opt}$.    To this end, we first compute a bound on the 
 Lyapunov drift. 
 
 \begin{lem} \label{lem:lyap-drift} 
 (Lyapunov drift bound) For all 
 $t$ and all possible values of $\bv{Q}(t)$, we have: 
 \begin{eqnarray*}
 \Delta(\bv{Q}(t)) &\leq& B - \sum_{n=1}^N (Q_n(t) - \theta_n)\expect{\mu_n(t) - A_n(t)\left|\right.\bv{Q}(t)} 
 \end{eqnarray*}
 where $B$ is a finite constant that satisfies: 
 \begin{eqnarray}
 B \geq \frac{1}{2}\sum_{n=1}^N \expect{(\mu_n(t) - A_n(t))^2 \left|\right.\bv{Q}(t)}  \label{eq:B}
  \end{eqnarray}
 Such a finite constant $B$ exists because of the boundedness assumptions 
 on buy and sell variables $\mu_n(t)$ and $A_n(t)$.  In particular, we have: 
 \begin{equation} \label{eq:B-ineq}  
 B  \leq  \frac{1}{2}\sum_{n=1}^N (\mu_{n}^{max})^2 
 \end{equation} 
 \end{lem}  
 \begin{proof}
 See Appendix A.
 \end{proof} 
 
 Using Lemma \ref{lem:lyap-drift} with the definition of $\phi(t)$ in (\ref{eq:phi}), a bound on the 
 drift-minus-reward expression is given as follows: 
 \begin{eqnarray}
 &\hspace{-1.3in}\Delta(\bv{Q}(t)) - V\expect{\phi(t)\left|\right.\bv{Q}(t)} \leq
 B  \nonumber \\
& - \sum_{n=1}^N(Q_n(t) - \theta_n)\expect{\mu_n(t) - A_n(t)\left|\right.\bv{Q}(t)} \nonumber \\
& -V\sum_{n=1}^N\expect{\mu_n(t) p_n(t) - s_n(\mu_n(t))\left|\right.\bv{Q}(t)} \nonumber \\
& +V\sum_{n=1}^N\expect{A_n(t)p_n(t) + b_n(A_n(t))\left|\right.\bv{Q}(t)}  \label{eq:drift-minus-penalty}
 \end{eqnarray} 
 We desire an algorithm that, every slot, observes the $\bv{Q}(t)$ values and the current
 prices, and makes a greedy trading action subject to the constraints (\ref{eq:mu1})-(\ref{eq:A2})
   that minimizes the right hand side of (\ref{eq:drift-minus-penalty}).   
 
 \subsection{The Dynamic Trading Algorithm} \label{section:alg} 
 
 Every slot $t$, observe $\bv{Q}(t)$ and $\bv{p}(t)$ and perform the following actions. 
 
 \begin{enumerate} 
 \item \emph{Selling:} For each $n \in \{1, \ldots, N\}$, choose $\mu_n(t)$ to solve: 
 \begin{eqnarray*}
 \mbox{Minimize:} & [\theta_n - Q_n(t)  -  Vp_n(t)]\mu_n(t) + Vs_n(\mu_n(t)) \\
 \mbox{Subject to:} & \mbox{Constraints (\ref{eq:mu1})-(\ref{eq:mu3})}
 \end{eqnarray*}
  
 \item \emph{Buying:} Choose $\bv{A}(t) = (A_1(t), \ldots, A_n(t))$ to solve: 
 \begin{eqnarray*}
 \mbox{Minimize:} & \sum_{n=1}^N[Q_n(t)-\theta_n + Vp_n(t)]A_n(t)  \\
 & + \sum_{n=1}^N Vb_n(A_n(t)) \\
 \mbox{Subject to:} & \mbox{Constraints (\ref{eq:A1})-(\ref{eq:A2})}
 \end{eqnarray*}  
 \end{enumerate} 
 
 The buying algorithm uses the integer constraints  (\ref{eq:A1})-(\ref{eq:A2}), and is
 related to the well known \emph{bounded
 knapsack problem} (it is exactly the bounded knapsack problem if the $b_n(\cdot)$ functions are linear).   
 Implementation of this integer constrained problem can be complex when the number of stocks $N$ is
 large.  
 However, if we use $x \defequiv \sum_{n=1}^N[\mu_n^{max}p_n^{max} + b_n(\mu_n^{max})]$, then
 constraint (\ref{eq:A2}) is effectively removed.  In this case, the stocks are decoupled and the buying 
 algorithm reduces to making separate decisions for each stock $n$.  Alternatively, 
 the constraint (\ref{eq:A2}) can be replaced by the constraint (\ref{eq:new-constraint-x}).  In
 this case, it is easy to see that if buying costs are linear,  so that $b_n(A) = b_n A$ for all $n$ (for some positive
 constants $b_n$), then
 the buying algorithm 
 reduces to successively
 buying as much stock as possible from the queues with the smallest (and negative) $[Q_n(t)-\theta_n + V(p_n(t)+b_n)]$ values. 
An alternative relaxation of the constraint (\ref{eq:A2}) is discussed in Section \ref{section:relax}.  
 
 \begin{lem} \label{lem:minimizes-drift} For a given $\bv{Q}(t)$ on slot $t$, 
  the above dynamic trading algorithm satisfies: 
  \begin{eqnarray}
  B -V\phi(t) - \sum_{n=1}^N(Q_n(t)-\theta_n)(\mu_n(t) - A_n(t))  \leq \nonumber \\
  B -V\phi^*(t) - \sum_{n=1}^N(Q_n(t) - \theta_n)(\mu_n^*(t) - A_n^*(t))  \label{eq:min1} 
  \end{eqnarray}
  where $\bv{A}(t)$, $\bv{\mu}(t)$ are the actual 
  decisions made by the algorithm, which define $\phi(t)$ by (\ref{eq:phi}), and $\bv{A}^*(t)$, $\bv{\mu}^*(t)$ are any
  alternative (possibly randomized) 
  decisions that can be made on slot $t$ that satisfy (\ref{eq:mu1})-(\ref{eq:A2}), which define $\phi^*(t)$ 
  by (\ref{eq:phi}). Furthermore, we have: 
   \begin{eqnarray}
 &\hspace{-1.3in}\Delta(\bv{Q}(t)) - V\expect{\phi(t)\left|\right.\bv{Q}(t)} \leq
 B  \nonumber \\
& - \sum_{n=1}^N(Q_n(t) - \theta_n)\expect{\mu_n^*(t) - A_n^*(t)\left|\right.\bv{Q}(t)} \nonumber \\
& -V\sum_{n=1}^N\expect{\mu_n^*(t) p_n(t) - s_n(\mu_n^*(t))\left|\right.\bv{Q}(t)} \nonumber \\
& +V\sum_{n=1}^N\expect{A_n^*(t)p_n(t) + b_n(A_n^*(t))\left|\right.\bv{Q}(t)}  \label{eq:drift-minus-penalty2}
 \end{eqnarray} 
 where the expectation on the right hand side of (\ref{eq:drift-minus-penalty2}) is with respect to the 
 random price vector $\bv{p}(t)$ and the possibly random actions $\bv{A}^*(t)$, $\bv{\mu}^*(t)$ in response
 to this price vector. 
 \end{lem}
  \begin{proof} 
Given $\bv{Q}(t)$ on slot $t$, 
the dynamic algorithm makes buying and selling decisions to minimize the left hand side of 
(\ref{eq:min1}) over all alternative decisions that satisfy (\ref{eq:mu1})-(\ref{eq:A2}).  Therefore, 
the inequality (\ref{eq:min1}) holds for all realizations of the random quantities, and 
hence also holds when taking conditional expectations of both sides. The
conditional expectation of the left hand side of (\ref{eq:min1}) is equivalent to the right hand side 
of the drift-minus-reward expression (\ref{eq:drift-minus-penalty}), which proves (\ref{eq:drift-minus-penalty2}). 
\end{proof} 

The main idea behind our analysis is that the Dynamic Trading Algorithm is simple to implement and 
does not require knowledge of the future or of the statistics of the price process $\bv{p}(t)$.  
However, it can be \emph{compared} to alternative policies
$\bv{A}^*(t)$ and $\bv{\mu}^*(t)$ (such as in Lemma \ref{lem:minimizes-drift}, and in other lemmas
in Sections \ref{section:non-iid} and \ref{section:arbitrary-prices} that consider more complex price  processes),
and these policies 
possibly have knowledge both of the price statistics and of the future.  
 
 \subsection{Bounding the Stock Queues} 
 
 The next lemma shows that the above algorithm does not sell any shares of stock $n$ if $Q_n(t)$ is sufficiently
 small.
   
 \begin{lem} \label{lem:lower-bound} Under the above Dynamic Trading Algorithm and for arbitrary price 
 processes $\bv{p}(t)$ that satisfy (\ref{eq:price-bound}), 
 if $Q_n(t) < \theta_n - Vp_n^{max}$ for some particular queue $n$ and slot $t$, then $\mu_n(t) = 0$.   Therefore,  
 if $Q_n(0) \geq \theta_n - Vp_n^{max} - \mu_n^{max}$, then: 
 \[ Q_n(t) \geq \theta_n - Vp_n^{max} - \mu_n^{max} \: \: \mbox{ for all $t$} \]
 \end{lem} 
 \begin{proof} 
 Suppose that $Q_n(t) < \theta_n - Vp_n^{max}$ for some particular queue $n$ and slot $t$.
 Then for any $\mu\geq 0$ we have: 
\begin{eqnarray*}
&&\hspace{-.9in} [\theta_n - Q_n(t) - Vp_n(t)]\mu + V s_n(\mu) \\
\hspace{+.4in}&\geq&   [\theta_n - Q_n(t) - Vp_n^{max}]\mu  + V s_n(\mu) \\
\hspace{+.4in}&\geq& [\theta_n - Q_n(t) - Vp_n^{max}]\mu \\
\hspace{+.4in}&\geq& 0
 \end{eqnarray*}
 where the final inequality holds with equality if and only if $\mu=0$. 
 Therefore, the Dynamic Trading Algorithm must choose $\mu_n(t) = 0$. 
 
 Now suppose that $Q_n(t) \geq \theta_n - Vp_n^{max} - \mu_n^{max}$ for some time $t$. We show it also holds for $t+1$.
 If $Q_n(t) \geq \theta_n - Vp_n^{max}$, then it can decrease by at most $\mu_n^{max}$ on a single slot, so that 
 $Q_n(t+1) \geq \theta_n - Vp_n^{max} - \mu_n^{max}$.  
 Conversely, if $\theta_n  - Vp_n^{max} > Q_n(t) \geq \theta_n - Vp_n^{max} - \mu_n^{max}$, 
 then we know $\mu_n(t)=0$ and so the queue cannot decrease on the next slot and 
 we again have $Q_n(t+1) \geq \theta_n - Vp_n^{max} - \mu_n^{max}$.  It follows that this inequality is always upheld if it is satisfied
 at $t=0$. 
 \end{proof}

We note that the above lemma is a \emph{sample path} statement that holds for arbitrary (possibly non-ergodic) 
price processes. 
 The next lemma also deals with sample paths, and 
 shows that all queues have a finite maximum size $Q_n^{max}$.
 
 \begin{lem} \label{lem:upper-bound} 
 Under the above Dynamic Trading Algorithm and for arbitrary price processes
 $\bv{p}(t)$ that satisfy (\ref{eq:price-bound}), if $Q_n(t) > \theta_n$ for some particular queue $n$ and 
 slot $t$, 
 then $A_n(t) = 0$ and so 
 the queue cannot increase on the next slot.  It follows that if $Q_n(0) \leq \theta_n + \mu_n^{max}$, then: 
 \[ Q_n(t) \leq \theta_n + \mu_n^{max} \: \: \mbox{ for all $t$} \]
 \end{lem} 
 \begin{proof} 
 Suppose that $Q_n(t) > \theta_n$ for a particular queue $n$ and  slot $t$. 
 Let  $\bv{A}(t) = (A_1(t), \ldots, A_N(t))$ be a vector of buying decisions that solve the 
 optimization associated with the Buying algorithm on slot $t$, so that they minimize the
 expression: 
 \begin{equation} \label{eq:expression} 
 \sum_{m=1}^N [Q_m(t) - \theta_m + Vp_m(t)]A_m(t) + \sum_{m=1}^N Vb_m(A_m(t)) 
 \end{equation} 
 subject to (\ref{eq:A1})-(\ref{eq:A2}). 
 Suppose that $A_n(t)>0$ (we shall reach a contradiction).  Because the term
 $[Q_n(t) - \theta_n + Vp_n(t)]$ is \emph{strictly positive}, and because the $b_n(A)$ function is non-decreasing, 
 we can strictly reduce the value of the expression (\ref{eq:expression})  by changing $A_n(t)$ to $0$. 
 This change still satisfies the constraints (\ref{eq:A1})-(\ref{eq:A2}) and produces a strictly smaller
 sum in (\ref{eq:expression}), 
 contradicting the assumption that $\bv{A}(t)$ is a  minimizer.
 Thus, if $Q_n(t) > \theta_n$, then $A_n(t) = 0$. 
 
 Because the queue value can increase by at most $\mu_n^{max}$ on any slot, and cannot increase if it already
 exceeds $\theta_n$,  it follows that 
 $Q_n(t) \leq \theta_n + \mu_n^{max}$ for all $t$, provided that this inequality holds at $t=0$. 
 \end{proof}

 \subsection{Analyzing Time Average Profit}  
 \begin{thm} \label{thm:1}   Fix any value $V>0$, and  define $\theta_n$ as follows: 
 \begin{equation} 
 \theta_n \defequiv Vp_n^{max} + 2\mu_n^{max} \label{eq:theta-n} 
 \end{equation} 
 Suppose that initial stock queues satisfy: 
 \begin{equation} \label{eq:initial-condition} 
 \mu_n^{max} \leq Q_n(0) \leq Vp_n^{max} + 3\mu_n^{max} 
 \end{equation} 
 If the Dynamic Trading Algorithm is implemented over $t \in \{0, 1, 2, \ldots\}$, then: 
 
 (a) Stock queues $Q_n(t)$ (for $n \in \{1, \ldots, N\}$) are deterministically bounded for all slots $t$
 as follows: 
 \begin{equation} \label{eq:q-bound} 
  \mu_n^{max} \leq Q_n(t) \leq Vp_n^{max} + 3\mu_n^{max} \: \: \mbox{ for all $n$ and all $t$} 
  \end{equation} 
 
 (b)  If $\bv{p}(t)$ is i.i.d. over slots with general distribution $Pr[\bv{p}(t) = \bv{p}] = \pi(\bv{p})$ for all
 $\bv{p} \in \script{P}$, then for all $t \in \{1, 2, \ldots\}$ we have: 
 \begin{equation} \label{eq:performance-bound} 
 \overline{\phi}(t) \geq \phi^{opt} - \frac{B}{V} - \frac{\expect{L(\bv{Q}(0))}}{Vt} 
 \end{equation} 
 where the constant $B$ is defined by (\ref{eq:B}) (and satisfies the inequality (\ref{eq:B-ineq})), 
 $\phi^{opt}$ is the optimal time average profit, and $\overline{\phi}(t)$ is the time average expected
 profit over $t$ slots: 
 \begin{eqnarray}
 &\overline{\phi}(t) \defequiv \frac{1}{t}\sum_{\tau=0}^{t-1} \expect{\phi(\tau)}& \label{eq:phi-t} 
 \end{eqnarray}
 Therefore: 
 \begin{equation} \label{eq:limit-performance} 
 \liminf_{t\rightarrow\infty} \overline{\phi}(t) \geq \phi^{opt} - B/V 
 \end{equation} 
 \end{thm}

Theorem \ref{thm:1} shows that the time average expected profit is within $B/V$ of the 
optimal value $\phi^{opt}$.  Because the $B$ constant is independent of $V$, we can choose
 $V$ to make $B/V$ arbitrarily small.  This comes with a tradeoff in the maximum  size required
 for each stock queue that is linear in $V$.  
 Specifically, the maximum stock level $Q_n^{max}$ required for stock $n$ is given as
 follows: 
\[ Q_n^{max} \defequiv Vp_n^{max} + 3\mu_n^{max} \]

Now suppose that we start with initial condition $Q_n(0) = \mu_n^{max}$ for all $n$ and all $t$. 
Then for $t \in \{1, 2, \ldots \}$ the error term $L(\bv{Q}(0))/(Vt)$ is given by: 
\begin{equation} \label{eq:timescales} 
 \frac{L(\bv{Q}(0))}{Vt} = \frac{\sum_{n=1}^N (Vp_n^{max }+ \mu_n^{max})^2}{2Vt} = O(V)/t 
 \end{equation} 
This shows that if $V$ is chosen to be large, then the amount of time $t$ required to make this error
term negligible must also be large.  One can minimize this error term with an initial condition 
$Q_n(0)$ that is close to $\theta_n$ for all $n$.  However, this is an artificial savings, as it does not include
the \emph{startup cost} associated with purchasing that many initial units of stock.    Therefore, the 
timescales are more accurately described by the transient given in (\ref{eq:timescales}).  

One may wonder how the Dynamic Trading Algorithm is achieving near optimal profit without 
knowing the distribution of the price vector $\bv{p}(t)$, and without estimating this distribution. 
The answer
is that it uses the queue values themselves to guide decisions.  These queue values $Q_n(t)$ 
only deviate significantly from the target $\theta_n$ when inefficient decisions are made.  The values then act as a
 ``sufficient statistic'' on which to base future decisions.   The same sufficient statistic holds for the non-i.i.d. case, as shown 
 in Section \ref{section:non-iid}, so that we do not need to estimate price patterns or time-correlations, provided that we allow
 for a sufficiently large control parameter $V$ and corresponding large timescales for convergence. 

Finally, one may also wonder if the limiting time average \emph{expected profit} given in  
(\ref{eq:limit-performance}) also holds (with probability 1) for the limiting time average profit (without
the expectation).  When $\bv{p}(t)$ evolves according to 
a finite state irreducible Markov chain (as is the case in this i.i.d. scenario), 
then the Dynamic Trading Algorithm in turn makes $\bv{Q}(t)$ evolve according to a finite state
Markov chain, and it can be shown that the limiting time average expected profit is the same (with probability 1) 
as the limiting time average profit.

 \subsection{Proof of Theorem \ref{thm:1}} 
  \begin{proof} (Theorem \ref{thm:1} part (a)) 
 By Lemma \ref{lem:lower-bound} we know that $Q_n(t) \geq \theta_n - Vp_n^{max} - \mu_n^{max}$ for all
 $t$ (provided that this holds at $t=0$).  However, $\theta_n - Vp_n^{max} - \mu_n^{max} = \mu_n^{max}$. 
 Thus, $Q_n(t) \geq \mu_n^{max}$ for all $t$, provided that this holds for $t=0$.   Similarly, 
by Lemma \ref{lem:upper-bound} we know that $Q_n(t) \leq \theta_n + \mu_n^{max}$ for all $t$ (provided
that this holds for $t=0$), and $\theta_n + \mu_n^{max} = Q_{n}^{max}$. 
 \end{proof} 
 
\begin{proof} (Theorem \ref{thm:1} part (b)) 
Fix a slot $t \in \{0, 1, 2, \ldots\}$. 
To prove part (b), we plug an alternative set of control choices 
$\bv{A}^*(t)$ and $\bv{\mu}^*(t)$ into the drift-minus-reward bound
(\ref{eq:drift-minus-penalty2}) of Lemma \ref{lem:minimizes-drift}. 
Because $\bv{p}(t)$ is i.i.d., we can choose $\bv{A}^*(t)$ and $\bv{\mu}^*(t)$ as the $p$-only policy 
 that satisfies (\ref{eq:zero-drift-iid}), (\ref{eq:opt-profit-iid}). 
 Note that we must first ensure this 
 $p$-only policy satisfies the constraint (\ref{eq:mu3}) 
 needed to apply the bound (\ref{eq:drift-minus-penalty2}).  However, we know from
 part (a) of this theorem that $Q_n(t) \geq \mu_n^{max}$ for all $n$, and so the constraint
 (\ref{eq:mu3}) is \emph{trivially} satisfied.  Therefore, we can plug
 this policy $\bv{A}^*(t)$ and $\bv{\mu}^*(t)$ into (\ref{eq:drift-minus-penalty2}) and use
 equalities (\ref{eq:zero-drift-iid}) and (\ref{eq:opt-profit-iid}) to yield: 
   \begin{eqnarray*}
 \Delta(\bv{Q}(t)) - V\expect{\phi(t)\left|\right.\bv{Q}(t)} \leq
 B  - V\phi^{opt}  
 \end{eqnarray*} 
 Taking expectations of the above inequality over the distribution of $\bv{Q}(t)$ and using the law of iterated
 expectations yields: 
 \[ \expect{L(\bv{Q}(t+1)) - L(\bv{Q}(t))} - V\expect{\phi(t)} \leq B - V\phi^{opt} \]
 The above holds for all $t \in \{0, 1, 2, \ldots, \}$.  Summing the above over $\tau \in \{0, \ldots, t-1\}$ (for some
 positive integer $t$) yields: 
 \[ \expect{L(\bv{Q}(t)) - L(\bv{Q}(0))} - V\sum_{\tau=0}^{t-1} \expect{\phi(\tau)} \leq tB - tV\phi^{opt} \]
 Dividing by $tV$, rearranging terms, and using non-negativity of $L(\cdot)$ yields: 
 \[ \overline{\phi}(t) \geq \phi^{opt} - B/V - \expect{L(\bv{Q}(0))}/Vt \]
 where $\overline{\phi}(t)$ is defined in (\ref{eq:phi-t}). This proves the result. 
\end{proof}

\section{Non-I.I.D. Prices}  \label{section:non-iid} 

Here we consider a general class of non-i.i.d. price processes that have a mild \emph{decaying memory property}. 
We first note that the only place a change is needed is in the proof of Theorem \ref{thm:1} part (b).  Indeed, part (a) 
of Theorem \ref{thm:1} is a \emph{sample path statement} that is true for any $\bv{p}(t)$ process.  That is, regardless
of whether or not $\bv{p}(t)$ is i.i.d. over slots, and even if it does not have well defined time averages 
as in (\ref{eq:time-avg-p}), we still have: 
\[   \mu_n^{max} \leq Q_n(t) \leq Vp_n^{max} + 3\mu_n^{max} \: \: \mbox{for all $n$ and all $t$}  \]
provided that this inequality is upheld at time $0$, and that the $\theta_n$ values are 
defined as in (\ref{eq:theta-n}). 

\subsection{The Decaying Memory Property} 

First consider any price vector process $\bv{p}(t)$ that satisfies (\ref{eq:time-avg-p}), where $\pi(\bv{p})$ is the 
time average fraction of time that $\bv{p}(t) = \bv{p}$. 
Consider implementing the $p$-only policy that would achieve (\ref{eq:zero-drift-iid}) and (\ref{eq:opt-profit-iid})
on each slot $t$ if the process where i.i.d. with the same steady state distribution $\pi(\bv{p})$.  We call
this the \emph{optimal $p$-only policy}.  Let $\bv{A}^*(t)$
and $\bv{\mu}^*(t)$ represent the resulting decision variables under this policy.  
Because these decisions react only to the current
$\bv{p}(t)$, and because the limiting fraction of time of being in each price state is the same as the i.i.d. case, 
the identities (\ref{eq:zero-drift-iid}) and (\ref{eq:opt-profit-iid})  
are now true in the limit as $t \rightarrow \infty$ (rather than true on every slot $t$): 
\begin{eqnarray*}
0 &=& \lim_{t\rightarrow\infty} \frac{1}{t}\sum_{\tau=0}^{t-1} \expect{A_n^*(\tau) - \mu_n^*(\tau)} \: \: \mbox{ for all $n$} \\
\phi^{opt} &=& \lim_{t\rightarrow\infty} \frac{1}{t}\sum_{\tau=0}^{t-1}\expect{\phi^*(\tau)}
\end{eqnarray*}
where $\phi^*(\tau)$ is defined: 
\begin{eqnarray}
\phi^*(\tau) &\defequiv& \sum_{n=1}^N\expect{\mu_n^*(\tau) p_n(\tau) - s_n(\mu_n^*(\tau))} \nonumber \\
&& -  \sum_{n=1}^{N} \expect{A_n^*(\tau) p_n(\tau) + b_N(A_n^*(\tau))}  \label{eq:phi-star-tau} 
\end{eqnarray}

We now further assume that the $\bv{p}(t)$ process achieves time averages that are close to these
limits when summed over an interval of $T$ slots, regardless of the past history before the interval.  
Specifically,  let $H(t)$ denote the history of the system up to slot  $t$, defined: 
\[ H(t) \defequiv [\bv{Q}(t), \bv{Q}(t-1), \ldots, \bv{Q}(0); \bv{p}(t-1), \bv{p}(t-2), \ldots, \bv{p}(0)] \]
Assume there are arbitrarily small values
$\epsilon>0$ for which there exists a positive integer $T$ (that may depend on $\epsilon$) 
such that the optimal $p$-only policy yields the following:  For any 
slot $t_0 \in \{0, 1, 2, \ldots\}$ and any $H(t_0)$, 
we have for all $n \in \{1, \ldots, N\}$: 
\begin{eqnarray} 
\left|\frac{1}{T}\sum_{\tau=t_0}^{t_0+T-1}\expect{A_n^*(\tau) - \mu_n^*(\tau)\left|\right.H(t_0)} \right| \leq \epsilon \label{eq:drift-noniid} 
\end{eqnarray} 
and
\begin{eqnarray}
\left| \phi^{opt} - \frac{1}{T}\sum_{\tau=t_0}^{t_0 + T-1}\expect{\phi^*(\tau) \left|\right.H(t_0)} \right|\leq \epsilon \label{eq:near-profit} 
\end{eqnarray}

We say that the stochastic process $\bv{p}(t)$ has the \emph{decaying memory property} 
if it satisfies (\ref{eq:drift-noniid}) and (\ref{eq:near-profit}).  This property ensures 
that time averages over any interval of $T$ slots are uniformly close to their steady state values, 
regardless of past history.  The simplest model that satisfies this decaying memory property  is the \emph{i.i.d. model}, 
for which we can use $T=1$ and $\epsilon = 0$. However, the
decaying memory property  is also satisfied by any  $\bv{p}(t)$ process  
that evolves according to a finite state ergodic  Markov chain, where the integer $T$ is related
to the ``mixing time'' of the chain.

\subsection{Performance} 
\begin{thm}   \label{thm:2}  
Suppose the Dynamic Trading Algorithm is implemented, with $\theta_n$ values satisfying 
(\ref{eq:theta-n}), and initial condition
that satisfies (\ref{eq:initial-condition}).  Then the queue backlog satisfies the deterministic bound 
(\ref{eq:q-bound}).  Further, for any pair $T$, $\epsilon$ that satisfies (\ref{eq:drift-noniid}), (\ref{eq:near-profit}), 
we have for any integer $M \in \{1, 2, 3, \ldots\}$: 
\begin{equation} \label{eq:noniid1} 
 \overline{\phi}(MT) \geq \phi^{opt} - C_2\epsilon - C_1T/V - \frac{\expect{L(\bv{Q}(0))}}{VMT} 
 \end{equation} 
 and: 
\begin{equation} \label{eq:noniid2} 
  \liminf_{t\rightarrow\infty} \overline{\phi}(t) \geq \phi^{opt} - C_2\epsilon - C_1T/V       
  \end{equation} 
where $C_1$ and $C_2$ are defined: 
\begin{eqnarray*}
C_1 &\defequiv& \sum_{n=1}^N(\mu_n^{max})^2 \left[\frac{3}{2} + \frac{1}{T}  + \frac{1}{2T^2}\right] + \frac{\epsilon}{T}\sum_{n=1}^N\mu_n^{max} \\
C_2 &\defequiv& 1 + \sum_{n=1}^Np_n^{max} 
\end{eqnarray*}
If $\bv{Q}(0) = (\mu_1^{max}, \ldots, \mu_N^{max})$, 
then $L(\bv{Q}(0))/(VMT)$ has the form (\ref{eq:timescales}) with $t = MT$.
\end{thm} 
\begin{proof}
The theorem is proven by a Lyapunov drift argument over $T$-slot frames, and is given in 
Appendix B. 
\end{proof} 

Note that the \emph{same} Dynamic Trading Algorithm as in the i.i.d. case is used here, without
requiring knowledge of $\epsilon$ or $T$.   Indeed, the above performance bounds (\ref{eq:noniid1})
and (\ref{eq:noniid2}) hold for any $\epsilon$, $T$ pair that satisfies (\ref{eq:drift-noniid}) and (\ref{eq:near-profit}).
The bounds can thus be optimized over all such $\epsilon$, $T$ pairs. However, it suffices to note that such pairs
can be found for arbitrarily small values of $\epsilon$.  Thus, choosing a large value of $V$ 
makes achieved profit arbitrarily 
close to the optimal value $\phi^{opt}$.  However, if the $\bv{p}(t)$ process has a long 
``mixing time,''  then 
the value of $T$ needed for a given $\epsilon$ will 
be large, and so the $V$ parameter will also need to be chosen to be large.  Thus, non-i.i.d. 
$\bv{p}(t)$ processes typically require larger queue sizes to ensure close proximity to the optimal profit.

\section{Arbitrary Price Processes} \label{section:arbitrary-prices} 

Here we consider the performance of the Dynamic Trading Algorithm for an 
\emph{arbitrary} price vector process $\bv{p}(t)$, possibly a non-ergodic process without a well
defined time average such as that given in (\ref{eq:time-avg-p}). In this case, there may not be a well
defined ``optimal'' time average profit $\phi^{opt}$.  However, one can define $\phi^{opt}(t)$ as the 
maximum possible time average 
profit achievable over the interval $\{0, \ldots, t-1\}$ by an algorithm with perfect
knowledge of the future and that conforms to the constraints (\ref{eq:mu1})-(\ref{eq:A2}).   For the ergodic settings
 described in the previous sections, $\phi^{opt}(t)$ has a well defined limiting value, and 
 our algorithm comes close to its limiting value.  In this (possibly 
non-ergodic) setting, we do not 
claim that our algorithm comes close to $\phi^{opt}(t)$.  Rather, we make a less ambitious claim that 
our policy yields a profit that is close to (or greater than) 
the profit achievable by a frame-based policy that can look only $T$ slots
into the future.  

\subsection{The $T$-Slot Lookahead Performance} 

Let $T$ be a positive integer, and fix any slot $t_0 \in \{0, 1, 2, \ldots\}$. 
Define $\psi_T(t_0)$ as the optimal profit achievable over the interval $\{t_0, \ldots, t_0 + T-1\}$ 
by a policy that has perfect a-priori knowledge of the prices $\bv{p}(\tau)$ over this interval, and that 
ensures for each $n  \in \{1, \ldots, N\}$ that the total amount of stock $n$ 
purchased over this interval is greater than or equal to the total amount sold.   Specifically, $\psi_T(t_0)$ 
is mathematically defined according to the following optimization problem that has decision 
variables $\bv{A}(\tau)$, $\bv{\mu}(\tau)$, and that treats the stock prices
$\bv{p}(\tau)$ as deterministically known quantities:
\begin{eqnarray}
\mbox{Max:}&\psi \defequiv \sum_{\tau=t_0}^{t_0+T-1} \sum_{n=1}^N[\mu_n(\tau)p_n(\tau) - s_n(\mu_n(\tau))] \nonumber \\
& \hspace{-.2in}- \sum_{\tau=t_0}^{t_0+T-1} \sum_{n=1}^N[A_n(\tau)p_n(\tau) + b_n(A_n(\tau))] \label{eq:t-slot-lookahead1} \\
\mbox{Subj. to:}&\sum_{\tau=t_0}^{t_0+T-1} A_n(\tau) \geq \sum_{\tau=t_0}^{t_0+T-1} \mu_n(\tau) \: \forall n \label{eq:t-slot-lookahead2} \\
&\mbox{Constraints (\ref{eq:mu1}), (\ref{eq:mu2}), (\ref{eq:A1}), (\ref{eq:A2})} \label{eq:t-slot-lookahead3} 
\end{eqnarray}

The value $\psi_T(t_0)$ is equal to the maximizing value $\psi$ in the above 
problem (\ref{eq:t-slot-lookahead1})-(\ref{eq:t-slot-lookahead3}).   Note that the constraint (\ref{eq:t-slot-lookahead2}) 
only 
requires the amount of type-$n$ stock purchased to be greater than or equal to the amount sold by the end of the $T$-slot
interval, and does not require this at intermediate steps of the interval.  This allows the $T$-slot Lookahead policy 
to \emph{sell short}
stock that is not yet owned, provided that the requisite amount is purchased by the end of the interval.  

Note that the trivial decisions $\bv{A}(\tau) = \bv{\mu}(\tau) = \bv{0}$ for $\tau \in \{t_0, \ldots, t_0+T-1\}$ lead to $0$
profit over the interval, and hence $\Psi_T(t_0) \geq 0$ for all $T$ and all $t_0$. 
Consider now
the interval $\{0, 1, \ldots, MT-1\}$ that is divided into a total of $M$ frames of $T$-slots.  
We show that for any positive
integer $M$, our
Dynamic Trading Algorithm yields an average profit over this interval that is close to the average
profit of a $T$-slot lookahead policy that is implemented on each $T$-slot frame of this interval.

\subsection{The $T$-Slot Sample Path Drift} 

Let $L(\bv{Q}(t))$ be the Lyapunov function of (\ref{eq:lyap-function}). 
For a given slot $t$ and a given positive integer $T$, 
define the \emph{$T$-slot sample path drift} $\hat{\Delta}_T(t)$ as follows: 
\begin{equation} \label{eq:delta-hat} 
\hat{\Delta}_T(t) \defequiv L(\bv{Q}(t+T)) - L(\bv{Q}(t))
\end{equation} 
This differs from the one-slot conditional Lyapunov drift in (\ref{eq:one-slot-drift})  in two respects: 
\begin{itemize} 
\item It considers the difference in the Lyapunov function over $T$ slots, rather than a single slot.
\item It is a 
random variable equal to the difference between the Lyapunov function on slots $t$ and $t+T$, 
rather than a conditional expectation of this difference. 
\end{itemize} 

\begin{lem}   \label{lem:baba-sample} 
Suppose the Dynamic Trading Algorithm is implemented, with $\theta_n$ values satisfying 
(\ref{eq:theta-n}), and initial condition
that satisfies (\ref{eq:initial-condition}).  Then for any 
given slot $t_0$ and all integers $T>0$,  we have: 
\begin{eqnarray*}
&\hspace{-.2in}\hat{\Delta}_T(t_0) - V\sum_{\tau=t_0}^{t_0+T-1} \phi(\tau) \leq DT^2 - V\sum_{\tau=t_0}^{t_0+T-1} \phi^*(\tau) \nonumber \\
&+ \sum_{n=1}^N|Q_n(t_0) - \theta_n|\sum_{\tau=t_0}^{t_0+T-1} [\mu_n^*(\tau) - A_n^*(\tau)]\nonumber 
\end{eqnarray*}
where $\phi(\tau)$ is defined in (\ref{eq:phi}), and $\phi^*(\tau)$, $\bv{\mu}^*(\tau)$, 
$\bv{A}^*(\tau)$ represent any alternative control actions for slot $\tau$  that satisfy 
the constraints (\ref{eq:mu1}), (\ref{eq:mu2}), (\ref{eq:A1}), (\ref{eq:A2}).  Further, the constant $D$ is given by: 
\begin{equation} \label{eq:D}
D \defequiv \left[\frac{3}{2} + \frac{1}{2T^2} + \frac{1}{T}\right]\sum_{n=1}^N(\mu_n^{max})^2
\end{equation} 
\end{lem} 
\begin{proof} 
This lemma is identical to  Lemma \ref{lem:baba} in Appendix B, and the proof is given there. 
\end{proof} 

 \begin{thm} \label{thm:3} 
Suppose the Dynamic Trading Algorithm is implemented, with $\theta_n$ values satisfying 
(\ref{eq:theta-n}), and initial condition
that satisfies (\ref{eq:initial-condition}).  Then for any arbitrary price process $\bv{p}(t)$ that 
satisfies (\ref{eq:price-bound}), we have: 

(a) All queues $Q_n(t)$ are bounded according to (\ref{eq:q-bound}). 

(b) For any positive integers $M$ and $T$, the time average profit over
the interval $\{0, \ldots, MT-1\}$ satisfies the deterministic bound: 
\begin{eqnarray} 
 \frac{1}{MT}\sum_{\tau=0}^{MT-1} \phi(\tau) &\geq& \frac{1}{MT}\sum_{m=0}^{M-1} \psi_T(mT)  \nonumber \\
 && - \frac{DT}{V} - \frac{L(\bv{Q}(0))}{MTV} \label{eq:non-ergodic-performance} 
 \end{eqnarray}
where the $\psi_T(mT)$ values are defined according to the $T$-slot lookahead policy 
that uses knowledge of the 
future to solve (\ref{eq:t-slot-lookahead1})-(\ref{eq:t-slot-lookahead3}) for each $T$-slot frame. 
The constant $D$ is defined in (\ref{eq:D}), and if $\bv{Q}(0) = (\mu_1^{max}, \ldots, \mu_N^{max})$
then $L(\bv{Q}(0))/(MTV)$ has the form (\ref{eq:timescales}) with $t = MT$.
 \end{thm} 
 \begin{proof} 
 Part (a) has already been proven in Theorem \ref{thm:1}.  To prove part (b), fix any slot $t_0$ and any positive integer $T$. 
 Define $\bv{A}^*(\tau)$ and $\bv{\mu}^*(\tau)$
 as the solution of (\ref{eq:t-slot-lookahead1})-(\ref{eq:t-slot-lookahead3}) over the interval
 $\tau \in \{t_0, \ldots, t_0+T-1\}$.  By (\ref{eq:t-slot-lookahead3}),  these decision variables  satisfy constraints
 (\ref{eq:mu1}), (\ref{eq:mu2}), (\ref{eq:A1}), (\ref{eq:A2}), and hence can be plugged in to the bound in 
 Lemma \ref{lem:baba-sample}.  Because (\ref{eq:t-slot-lookahead1}), (\ref{eq:t-slot-lookahead2}) hold
 for these variables, by Lemma \ref{lem:baba-sample} we have: 
   \[ \hat{\Delta}_T(t_0)  - V\sum_{\tau=t_0}^{t_0+T-1} \phi(\tau) \leq DT^2 - V\psi_T(t_0) \]
Using the definition of $\hat{\Delta}_T(t_0)$ given in (\ref{eq:delta-hat}) yields:  
  \[ L(\bv{Q}(t_0+T)) - L(\bv{Q}(t_0)) - V\sum_{\tau=t_0}^{t_0+T-1} \phi(\tau) \leq DT^2 - V\psi_T(t_0) \]
The above inequality holds for all slots $t_0 \in \{0, 1, 2, \ldots\}$.  Letting $t_0 = mT$ and summing
over $m \in \{0, 1, \ldots, M-1\}$ (for some positive integer $M$) yields: 
\begin{eqnarray*}
L(\bv{Q}(MT)) - L(\bv{Q}(0)) - V\sum_{\tau=0}^{MT-1} \phi(\tau) \leq \\
MDT^2 - V\sum_{m=0}^{M-1} \psi_T(mT)
\end{eqnarray*}
Rearranging terms and using non-negativity of $L(\cdot)$ proves the theorem.
 \end{proof}

 Theorem \ref{thm:3} is stated for general price processes, but has explicit performance bounds for queue
 size in terms of the chosen $V$ parameter, and for profit in terms of $V$ and of the profit $\psi_T(mT)$ 
 of $T$-slot lookahead
 policies.    Plugging a large value of $T$ into the bound (\ref{eq:non-ergodic-performance}) increases the 
 first term on the right hand side because it allows for a larger amount of lookahead.  However, this comes
 with the cost of increasing the term $DT/V$ that is required to be small to ensure close proximity to the desired
 profit.  One can use this theorem with any desired model of stock prices to compute statistics
 associated with $\psi_T(mT)$ and hence understand more precisely 
 the timescales over which near-optimal profit is achieved.

 \section{Place-Holder Stock} \label{section:place-holder}

Theorems \ref{thm:1}, \ref{thm:2}, \ref{thm:3}  require an initial stock level of at least $\mu_n^{max}$ in all of the $N$ stocks.  This can 
be achieved by initially purchasing these shares (say, at time $t = -1$).   This creates an initial
startup cost that, while independent of $V$, can still be substantial.  It turns out that we can achieve the \emph{same}
performance as specified in Theorems \ref{thm:1}, \ref{thm:2}, \ref{thm:3}
without paying this startup cost. This can be done using the concept of
\emph{place-holder backlog} from \cite{neely-asilomar08}, which becomes \emph{place-holder stock} in our context. 

Specifically, suppose that we use $\hat{\bv{Q}}(t)$ to represent the \emph{actual} amount of stock held on slot $t$, 
and assume that $\hat{\bv{Q}}(0)$ satisfies: 
\[ 0 \leq \hat{Q}_n(0) \leq Vp_n^{max} + 2\mu_n^{max} \: \: \mbox{for all $n\in\{1, \ldots, N\}$}  \]
Define $\bv{Q}(t) \defequiv \hat{\bv{Q}}(t) + \bv{\mu}^{max}$ as an \emph{augmented stock vector}, 
where vector $\bv{\mu}^{max}$ is given by:
\[ \bv{\mu}^{max} \defequiv (\mu_1^{max}, \ldots, \mu_N^{max}) \]
Notice that the initial value of $\bv{Q}(0)$ satisfies (\ref{eq:initial-condition}). 
  Let us implement the 
Dynamic Trading Algorithm using the augmented stock vector $\bv{Q}(t)$.  This is equivalent to starting out the 
system with an initial amount that includes  $\mu_n^{max}$ \emph{fake shares} of stock in all queues.   
We then run the algorithm
on the $\bv{Q}(t)$ values, 
and any time we are asked to sell stock, we choose to sell \emph{real shares} whenever possible.  The algorithm 
breaks if at any time we are asked to sell at a level that is more than the number of real shares we have.  However, 
because on every sample path, we have $Q_n(t) \geq \mu_n^{max}$, we know that we are \emph{never} asked to sell
more real shares than we actually have.  Thus, these fake shares
simply act as \emph{place holders} to achieve the performance that would be achieved if we started
out with $\mu_n^{max}$ units of real shares in all queues.  Specifically, we 
achieve  performance guarantees specified in Theorems \ref{thm:1}, \ref{thm:2}, \ref{thm:3} 
associated with $\bv{Q}(0)$. If all actual queues are initially empty, then we have $\bv{Q}(0) = \bv{\mu}^{max}$, 
and hence we also have transients corresponding to $L(\bv{Q}(0))=L(\bv{\mu}^{max})$, 
without having to pay the startup cost of purchasing $\mu_n^{max}$ shares of each stock. 
 
 \section{Extensions} \label{section:splits}
 
 \subsection{Price Jumps and Stock Splits} 
 We have assumed that prices are bounded by values $p_n^{max}$ for simplicity of exposition.  In practice, 
 the $p_n^{max}$ values can be chosen as  price levels that we do not expect to see (perhaps 3 or 4 times the current price).  
 The prediction should be small enough to maintain reasonably small values for
 $\theta_n$ and $Q_n^{max}$, given in 
 (\ref{eq:theta-n}) and (\ref{eq:q-bound}).
 
 In the (desirable) situation when the price of a certain stock $n$ exceeds our estimated upper bound 
 $p_n^{max}$, we can simply
 adjust $p_n^{max}$ to a higher value.  We must then also appropriately 
 adjust $\theta_n$ according to (\ref{eq:theta-n}).  This can be viewed as if we are starting the system off with a new
 initial condition at this time (given by the current queue state), with new parameter choices. Because Theorems \ref{thm:1}, 
 \ref{thm:2},  \ref{thm:3} are stated in terms of general initial conditions, the achieved performance is then also determined
 by these theorems (applied to the time interval starting at the current time).  Intuitively, this will not ``break'' the algorithm
 because it continuously adapts to emerging conditions. 
 
 Similarly, we might have a price go so high as to affect a stock split.  This (desirable) situation 
 can either be modeled by an increase in the $p_n^{max}$ value (maintaining the same number of 
 shares, but treating each share as being worth double the market price), or by doubling the number
 of shares of that stock and increasing the $\mu_n^{max}$ and/or the $V$ parameter to allow for more shares
 to be maintained.  Again, the new situation can be viewed as creating a new initial condition, and so the algorithm
 can adapt to such events. 
 
 \subsection{Scaling for Exponential Growth} 
 
 Suppose we run the Dynamic Trading Algorithm over a fixed window of $W$ slots, using parameters $\mu_n^{max}$ and $V$, 
 with $\theta_n^{max}$ defined by (\ref{eq:theta-n}).   Assume we use place-holder stock so that the actual 
 stock queues are $0$ at the beginning of the time window.  If the achieved profit over this window is $z$, then for
any given value $\alpha>0$, 
 a profit $(1+\alpha)z$ could have been achieved if we had scaled the 
  $\mu_n^{max}$ and $V$ parameters (and hence $\theta_n^{max}$ by (\ref{eq:theta-n})) 
  by a factor $(1+\alpha)$ (for simplicity, we ignore integer constraints in the scaling of $\mu_n^{max}$ for the
  high level discussion of this subsection).  Of course, doing this would require a tolerance to the extra amount of 
  risk associated with keeping that much more stock in the stock queues.  
  However, assuming our risk tolerance grows proportionally to our wealth, 
 this increased risk is tolerable on the  \emph{next} window of $W$ slots.   
 Specifically, choose a value $T$, and consider
 the $T$-slot lookahead policy for comparison using (\ref{eq:non-ergodic-performance}) of 
 Theorem \ref{thm:3}.    Fix a value $\epsilon>0$, and choose $\mu_n^{max}$, 
 $V$, and $M$ so that  $DT/V + L(\bv{\mu}^{max})/(MTV) \leq \epsilon$.  Let $W = MT$. Then by (\ref{eq:non-ergodic-performance})
 we know that time average profit over $W$ slots is within $\epsilon$ of that provided by the $T$-slot lookahead policy.

Now consider consecutive windows of $W$ slots, and define $q_w$ as the time average
profit that would be earned over the
$w$th window if we use place-holder stock with $0$ initial stock levels, and if we use parameters $\mu_n^{max}$, $V$, 
and $\theta_n^{max}$.  Let $q^{(T)}_w$ denote the time average profit of the $T$-slot lookahead policy over this same
window of time.  By Theorem \ref{thm:3} we have 
that $q_w \geq q^{(T)}_w - \epsilon$ for each window $w \in \{1, 2, \ldots\}$.   
Define $\alpha_w \defequiv \beta\max[q_w, 0]$, where $\beta$ is some
positive proportionality constant.  Then $\alpha_w$ is non-negative, and if it is positive then it is 
proportional to the profit earned over window $w$. 
On each window $w>1$, rather than using parameters  $\mu_n^{max}$, $V$, 
and $\theta_n^{max}$, we scale these by the following factor: 
\[ (1 + \alpha_1)(1+\alpha_2)\cdots(1+\alpha_{w-1}) \]
Ignoring integer constraints in this scaling for simplicity, 
we know that time average profit earned over window $w$  is 
at least: 
\[ (q^{(T)}_w - \epsilon)(1+\alpha_1)(1+\alpha_2)\cdots(1+\alpha_{w-1}) \]  
It follows that our wealth increases exponentially 
as $(1+\alpha_1)(1+\alpha_2)(1+\alpha_3)\ldots$, where the profit coefficients $\alpha_w$ are close
to those associated with the 
$T$-slot lookahead policy.   In particular, the $\alpha_i$ coefficients are all greater than a uniform positive
number whenever $q^{(T)}_w \geq 2\epsilon$ for all $w \in \{1, 2, \ldots\}$. 

\subsection{Relaxing the Buying Constraint (\ref{eq:A2})} \label{section:relax} 

The constraint (\ref{eq:A2}) can make the buying policy of the Dynamic Trading Algorithm difficult to implement when the 
number of stocks $N$ is large, 
as discussed after the description of the algorithm in Section \ref{section:alg}.    Here we consider a 
simple and greedy modification 
that \emph{relaxes} the constraint (\ref{eq:A2}):  
Assume the buying functions $b_n(A)$ are concave and non-decreasing. The algorithm seeks to minimize
the expression: 
\begin{equation} \label{eq:expression-overshoot} 
 \sum_{n=1}^N \left[ (Q_n(t) - \theta_n + Vp_n(t))A_n(t) +Vb_n(A_n(t))\right] 
 \end{equation} 
subject to $A_n(t) \in \{0, 1, \ldots, \mu_n^{max}\}$ for all $n \in \{1, \ldots, N\}$, and subject to
$\sum_{n=1}^NA_n(t)p_n(t) \leq x$.  Consider the following sequential algorithm for adding new shares until 
this last constraint is either met or exceeded:  Initialize $\bv{A} = (A_1, \ldots, A_N) = \bv{0}$.  On step $k$ of the procedure, 
for each $n\in\{1, \ldots, N\}$ such that $A_n \leq  \mu_n^{max}$, compute the value of: 
\[ \frac{(Q_n(t) - \theta_n + Vp_n(t)) + V(b_n(A_n+1) - b_n(A_n))}{p_n(t)} \]
If this value is non-negative for all $n \in \{1, \ldots, N\}$, stop and designate $\bv{A}(t) = \bv{A}$.  
Else, choose the $n$ with the smallest (negative) such value and add one more share to the $\bv{A}$ vector
in that entry $n$.  If the constraint $\sum_{n=1}^N A_n(t) p_n(t) \leq x$ is either met or exceeded, we are done and choose
$\bv{A}(t) = \bv{A}$.  Else, repeat the procedure with the new $\bv{A}$ vector. 

The intuition behind this greedy relaxation is that we choose to increment our allocation by one share in the stock with 
the smallest (negative) ratio given by the incremental change in (\ref{eq:expression-overshoot}) divided 
by the amount consumed
in the total money budget $x$.  This procedure
yields a vector $\bv{A}(t)$ that satisfies the constraints $A_n(t) \in \{0, 1, \ldots, \mu_n^{max}\}$
for all $n$, although it may violate the constraint $\sum_n A_n(t) p_n(t) \leq x$ by \emph{overshooting} the required value
$x$ with purchase of one extra share of a particular stock.  However, it has the property: 
\[ \sum_{n=1}^N A_n(t)p_n(t) \leq x + \max_{n\in\{1, \ldots, N\}} p_n^{max} \]
Therefore, we spend no more than a constant amount over our intended constraint $x$ on each slot. It can be shown 
that this greedy policy yields a value of the expression (\ref{eq:expression-overshoot}) that is less than or
equal to the corresponding expression that minimizes this value subject to the original constraints (\ref{eq:A1})-(\ref{eq:A2}). 
This is the key property used in Lemma \ref{lem:minimizes-drift} to prove Theorems \ref{thm:1}, \ref{thm:2}, \ref{thm:3}.  Hence, 
it can be shown that these theorems still hold under this relaxation.  Specifically, our queue sizes are still bounded according
to (\ref{eq:q-bound}) (which was derived using only the $\mu_n^{max}$ constraints and not constraint (\ref{eq:A2})), and our
time average profit (under this relaxed policy that does not necessarily satisfy (\ref{eq:A2}))
is close to or better than the corresponding policies used for comparison  in Theorems \ref{thm:1}, \ref{thm:2}, \ref{thm:3}, 
which \emph{do} satisfy the  constraint (\ref{eq:A2}).

\section{Conclusion} 

This work uses Lyapunov optimization theory, developed for stochastic optimization of queueing networks, 
to construct a dynamic policy for buying and selling stock.   
When prices are ergodic, a 
single non-anticipating policy was constructed
and shown to perform close to an ideal policy with perfect knowledge of the future, with a tradeoff in the 
required amount of stock kept in each queue and in the timescales associated with convergence. 
For arbitrary price sample paths, the same algorithm was shown to achieve a time average profit close
to that of a frame based $T$-slot lookahead policy that can look $T$ slots into the future.  
Our framework constrains the maximum number of stock shares that can be bought and sold at any 
time.  While this restricts the long term growth curve to a linear growth, it also limits risk by ensuring no
more than a constant value $Q_n^{max}$ shares of each stock $n$ are kept at any time.    A modified policy 
was briefly discussed that achieves exponential growth by scaling $Q_n^{max}$ in proportion to increased
risk tolerance as wealth increases.   These results add to the theory of universal stock trading, and 
are important for understanding optimal decision making in the presence of a complex and possibly
unknown price process.

 
 \section*{Appendix A --- Proof of Lemma \ref{lem:lyap-drift}} 
 
 Here we prove Lemma \ref{lem:lyap-drift}. 
 From the dynamics for $Q_n(t)$ in (\ref{eq:q-dynamics}) 
 we have: 
 \begin{eqnarray}
 &\hspace{-.25in}(Q_n(t+1) - \theta_n)^2 = (\max[Q_n(t) - \mu_n(t) + A_n(t), 0] - \theta_n)^2  \hspace{-.3in}\nonumber \\
 &\hspace{+.65in}\leq (Q_n(t) - \mu_n(t) + A_n(t) - \theta_n)^2 \label{eq:l11}
 \end{eqnarray}
 The inequality above holds because $\theta_n \geq 0$.  To see this, 
 note that 
 the inequality holds with equality if $Q_n(t) - \mu_n(t) + A_n(t) \geq 0$. 
 In the opposite case, the result of the $\max[\cdot, 0]$ operation is $0$, 
 and we have: 
 \[ (0 - \theta_n)^2 \leq (z - \theta_n)^2 \]
 where $z$ is any negative number, and so: 
 \[ (0 - \theta_n)^2 \leq (Q_n(t) - \mu_n(t) + A_n(t) - \theta_n)^2 \]
 
 From (\ref{eq:l11}) we have: 
 \begin{eqnarray*}
 \frac{(Q_n(t+1)-\theta_n)^2}{2} &\leq& \frac{(Q_n(t)-\theta_n)^2}{2}  + \frac{(\mu_n(t) - A_n(t))^2}{2} \\
 && - (Q_n(t) - \theta_n)(\mu_n(t) - A_n(t)) 
   \end{eqnarray*}
   Summing over $n \in \{1, \ldots, N\}$ and taking conditional expectations proves that: 
 \begin{eqnarray*}
 \Delta(\bv{Q}(t)) &\leq& \frac{1}{2}\sum_{n=1}^N\expect{(\mu_n(t) - A_n(t))^2\left|\right.\bv{Q}(t)} \\
 && - \sum_{n=1}^N(Q_n(t) - \theta)\expect{\mu_n(t) - A_n(t)\left|\right.\bv{Q}(t)} 
 \end{eqnarray*}
 Using the definition of $B$ in (\ref{eq:B}) to replace the first term on the right hand side above
 yields the result.

\section*{Appendix B --- Proof of Theorem \ref{thm:2}} 

\subsection{$T$-Slot Drift Analysis}

For the same Lyapunov function given in (\ref{eq:lyap-function}), and for a given 
positive integer $T$, 
define 
the \emph{$T$-slot conditional Lyapunov drift} as follows: 
\begin{eqnarray}
\Delta_T(H(t)) \defequiv \expect{L(\bv{Q}(t+T)) - L(\bv{Q}(t)) \left|\right.H(t)}  \label{eq:T-slot-drift} 
\end{eqnarray}
where $H(t)$ is the past history up to time $t$, 
defined as $[\bv{Q}(t), \bv{Q}(t-1), \ldots, \bv{Q}(0);  \bv{p}(t-1), \bv{p}(t-2), \ldots, \bv{p}(0)]$.
Also define the \emph{$T$-slot sample path drift} $\hat{\Delta}_T(t)$ as: 
\[ \hat{\Delta}_T(t) \defequiv L(\bv{Q}(t+T)) - L(\bv{Q}(t)) \]
With this definition, $\hat{\Delta}_T(t)$ is a random variable representing the difference between the Lyapunov
function at time $t+T$ and time $t$, and: 
\begin{equation} \label{eq:sample-to-expect} 
\expect{\hat{\Delta}_T(t)\left|\right.H(t)} = \Delta_T(H(t)) 
\end{equation} 
\begin{lem} \label{lem:T-slot} 
Suppose the Dynamic Trading Algorithm is implemented, with $\theta_n$ values satisfying 
(\ref{eq:theta-n}), and initial condition
that satisfies (\ref{eq:initial-condition}).  Then for all $t_0\in\{0, 1, 2, \ldots\}$, 
all integers $T>0$, and all possible values of $\bv{Q}(t_0)$ we have: 
\begin{eqnarray*}
\hat{\Delta}_T(t_0) \leq T^2\tilde{B}  - \sum_{n=1}^N(Q_n(t_0) - \theta_n)\sum_{\tau=t_0}^{t_0+T-1} [\mu_n(\tau) - A_n(\tau)]
\end{eqnarray*}
where $\tilde{B}$ is defined: 
\[ \tilde{B} \defequiv \frac{(1+1/T^2)}{2}\sum_{n=1}^N(\mu_n^{max})^2  \] 
\end{lem} 
\begin{proof} 
First note that: 
\begin{eqnarray}
&\hspace{-1.4in}(Q_n(t_0+T) - \theta_n)^2 \leq (\mu_n^{max})^2 \nonumber \\
&+ \left(Q_n(t_0) - \sum_{\tau=t_0}^{t_0+T-1}[ \mu_n(\tau) - A_n(\tau)] - \theta_n\right)^2 \label{eq:noniid-square} 
\end{eqnarray}
This can be seen as follows:  If $Q_n(t_0+T) \geq \theta_n$, then by (\ref{eq:q-bound}) and (\ref{eq:theta-n})
we know that $|Q_n(t_0 + T) - \theta_n| \leq \mu_n^{max}$, and so the square of this quantity
is bounded by the first term on the right hand side of (\ref{eq:noniid-square}), so that (\ref{eq:noniid-square})
holds in this case. Else, suppose that  $Q_n(t_0 + T) < \theta_n$.  We then have: 
\[ \theta_n > Q_n(t_0+T) \geq Q_n(t_0) - \sum_{\tau=t_0}^{t_0+T-1}[ \mu_n(\tau) - A_n(\tau)] \]
where the second inequality 
holds because the right hand side neglects the $\max[\cdot, 0]$ in the queueing dynamics (\ref{eq:q-dynamics}). 
It follows that (\ref{eq:noniid-square}) again holds. 

From (\ref{eq:noniid-square}) we have: 
\begin{eqnarray*}
&\frac{1}{2}\left[(Q_n(t_0+T) - \theta_n)^2 - (Q_n(t_0) - \theta_n)^2\right] \leq (\mu_n^{max})^2/2 \\
&+ \frac{1}{2}\left(\sum_{\tau=t_0}^{t_0+T-1} [\mu_n(\tau) - A_n(\tau)]\right)^2 \\
&- (Q_n(t_0) - \theta_n)\sum_{\tau=t_0}^{t_0+T-1} [\mu_n(\tau) -A_n(\tau)]
\end{eqnarray*}
Note that $|\mu_n(\tau) - A_n(\tau)| \leq \mu_n^{max}$ for all $\tau$. 
Summing the above over $n \in \{1, \ldots, N\}$ yield the result. 
\end{proof}

\begin{lem} \label{lem:foo} Suppose the Dynamic Trading Algorithm is implemented, with $\theta_n$ values satisfying 
(\ref{eq:theta-n}), and initial condition
that satisfies (\ref{eq:initial-condition}).  Then for any times $\tau$
and $t_0$ such that $\tau \geq t_0$, 
and for any given $\bv{Q}(\tau)$, $\bv{Q}(t_0)$,  we
have: 
\begin{eqnarray*}
-V\phi(\tau) - \sum_{n=1}^N(Q_n(t_0) - \theta_n)(\mu_n(\tau) - A_n(\tau)) \leq  \\
 2|\tau-t_0|\sum_{n=1}^N(\mu_n^{max})^2 \\
-V\phi^*(\tau) - \sum_{n=1}^N(Q_n(t_0) -\theta_n)(\mu_n^*(\tau) - A_n^*(\tau)) 
\end{eqnarray*}
where $\phi(\tau)$ is defined in (\ref{eq:phi}), and $\phi^*(\tau)$, $\bv{\mu}^*(\tau)$, 
$\bv{A}^*(\tau)$ represent any alternative control actions for slot $\tau$  that satisfy 
the constraints (\ref{eq:mu1}), (\ref{eq:mu2}), (\ref{eq:A1}), (\ref{eq:A2}). 
\end{lem} 
\begin{proof} 
Because each queue can change by at most $\mu_n^{max}$ per slot, 
we have for each $n \in \{1, \ldots, N\}$: 
\begin{eqnarray}
 -Q_n(t_0) (\mu_n(\tau) - A_n(\tau)) \leq -Q_n(\tau)(\mu_n(\tau) - A_n(\tau))  \nonumber \\
+ |\tau - t_0|(\mu_n^{max})^2  \label{eq:foo0} 
\end{eqnarray}
Therefore: 
\begin{eqnarray}
&\hspace{-.6in}-V\phi(\tau) - \sum_{n=1}^N(Q_n(t_0) - \theta_n)(\mu_n(\tau) - A_n(\tau)) \nonumber  \\ 
&\hspace{-1.0in}\leq |\tau-t_0|\sum_{n=1}^N(\mu_n^{max})^2  - V\phi(\tau) \nonumber \\
& - \sum_{n=1}^N(Q_n(\tau) -\theta_n)(\mu_n(\tau) - A_n(\tau)) \nonumber \\
&  \hspace{-1.0in}\leq |\tau-t_0|\sum_{n=1}^N(\mu_n^{max})^2  - V\phi^*(\tau) \nonumber \\
& - \sum_{n=1}^N(Q_n(\tau) -\theta_n)(\mu_n^*(\tau) - A_n^*(\tau)) \label{eq:foo} \\
& \hspace{-.9in}\leq 2|\tau-t_0|\sum_{n=1}^N(\mu_n^{max})^2  - V\phi^*(\tau) \nonumber \\
& - \sum_{n=1}^N(Q_n(t_0) -\theta_n)(\mu_n^*(\tau) - A_n^*(\tau)) \label{eq:foo2} 
\end{eqnarray}
where (\ref{eq:foo}) holds because, from Lemma \ref{lem:minimizes-drift}, we know 
the Dynamic Trading Algorithm on slot $\tau$ 
minimizes the left hand side of the inequality
over all alternative decisions for slot $\tau$ that satisfy the constraints
(\ref{eq:mu1}), (\ref{eq:mu2}), (\ref{eq:A1}), (\ref{eq:A2}) (note that we already know $Q_n(\tau) \geq \mu_n^{max}$ and 
so constraint (\ref{eq:mu3}) is redundant). Inequality (\ref{eq:foo2}) follows by an argument
similar to (\ref{eq:foo0}). 
\end{proof} 

\begin{lem}   \label{lem:baba} 
Suppose the Dynamic Trading Algorithm is implemented, with $\theta_n$ values satisfying 
(\ref{eq:theta-n}), and initial condition
that satisfies (\ref{eq:initial-condition}).  Then for any 
given slot $t_0$, all integers $T>0$, and all possible values of $\bv{Q}(t_0)$ we have: 
\begin{eqnarray*}
&\hspace{-.2in}\hat{\Delta}_T(t_0) - V\sum_{\tau=t_0}^{t_0+T-1} \phi(\tau) \leq DT^2 - V\sum_{\tau=t_0}^{t_0+T-1}\phi^*(\tau) \nonumber \\
&+ \sum_{n=1}^N|Q_n(t_0) - \theta_n|\sum_{\tau=t_0}^{t_0+T-1}[\mu_n^*(\tau) - A_n^*(\tau)] \nonumber 
\end{eqnarray*}
where $\phi(\tau)$ is defined in (\ref{eq:phi}), and $\phi^*(\tau)$, $\bv{\mu}^*(\tau)$, 
$\bv{A}^*(\tau)$ represent any alternative control actions for slot $\tau$  that satisfy 
the constraints (\ref{eq:mu1}), (\ref{eq:mu2}), (\ref{eq:A1}), (\ref{eq:A2}).  Further, the constant $D$ is defined: 
\[ D \defequiv \tilde{B} + (1+1/T)\sum_{n=1}^N(\mu_n^{max})^2 \]
\end{lem} 
\begin{proof} 
Summing the result of Lemma \ref{lem:foo} over $\tau \in \{t_0, \ldots, t_0+T-1\}$ and using Lemma \ref{lem:T-slot}
yields: 
\begin{eqnarray}
&\hspace{-1.3in}\hat{\Delta}_T(t_0) - V\sum_{\tau=t_0}^{t_0+T-1} \phi(\tau) \leq T^2\tilde{B} \nonumber \\
&\hspace{-.2in}+\sum_{n=1}^N(\mu_n^{max})^2(T-1)T - V\sum_{\tau=t_0}^{t_0+T-1} \phi^*(\tau) \nonumber \\
&\hspace{-.1in}- \sum_{n=1}^N(Q_n(t_0) - \theta_n)\sum_{\tau=t_0}^{t_0+T-1}[\mu_n^*(\tau) - A_n^*(\tau)] \label{eq:ineq-first}
\end{eqnarray}
Now note that $-(Q_n(t) - \theta_n) = |Q_n(t_0) - \theta_n|$ if $Q_n(t_0) \leq \theta_n$.  Else, 
if $Q_n(t_0) > \theta_n$ then $Q_n(t_0) - \theta_n = |Q_n(t_0) - \theta_n| \leq \mu_n^{max}$ 
(by (\ref{eq:q-bound}) and (\ref{eq:theta-n})). 
Thus: 
\begin{eqnarray*}
-\sum_{n=1}^N(Q_n(t_0) - \theta_n)\sum_{\tau=t_0}^{t_0+T-1} [\mu_n^*(\tau) - A_n^*(\tau)] \\
= \sum_{n=1}^N|Q_n(t_0) - \theta_n|\sum_{\tau=t_0}^{t_0+T-1} [\mu_n^*(\tau) - A_n^*\tau)] \\
-2 \sum_{n\in \script{M}(t)}|Q_n(t_0) - \theta_n|\sum_{\tau=t_0}^{t_0+T-1} [\mu_n^*(\tau) - A_n^*(\tau)]
\end{eqnarray*} 
where $\script{M}(t)$ is the set of all $n \in \{1, \ldots, N\}$ such that $Q_n(t) > \theta_n$. 
The final term is bounded by $2T\sum_{n=1}^N(\mu_n^{max})^2$.  Thus: 
\begin{eqnarray*}
&&\hspace{-.3in}-\sum_{n=1}^N(Q_n(t_0) - \theta_n)\sum_{\tau=t_0}^{t_0+T-1} [\mu_n^*(\tau) - A_n^*(\tau)] \\
&&\hspace{-.3in}\leq \sum_{n=1}^N|Q_n(t_0) - \theta_n|\sum_{\tau=t_0}^{t_0+T-1} [\mu_n^*(\tau) - A_n^*\tau)] + 2T\sum_{n=1}^N(\mu_n^{max})^2
\end{eqnarray*} 
Using this in (\ref{eq:ineq-first}) yields the result. 
\end{proof}

\subsection{The Time Average Profit} 
If the system satisfies the requirements specified in Lemma \ref{lem:baba}, then we can take conditional expectations
of $\hat{\Delta}_T(t_0)$ to yield (from (\ref{eq:sample-to-expect})): 
\begin{eqnarray*}
&\hspace{-.1in}\Delta_T(H(t_0)) - V\sum_{\tau=t_0}^{t_0+T-1} \expect{\phi(\tau)\left|\right.H(t_0)}  \leq DT^2  \\
&- V\sum_{\tau=t_0}^{t_0+T-1}\expect{\phi^*(\tau)\left|\right.H(t_0)}  \nonumber \\
&\hspace{-.1in}+ \sum_{n=1}^N|Q_n(t_0) - \theta_n|\sum_{\tau=t_0}^{t_0+T-1}\expect{\mu_n^*(\tau) - A_n^*(\tau)\left|\right.H(t_0)} \nonumber 
\end{eqnarray*}
Plugging the policy $\bv{A}^*(t)$, $\bv{\mu}^*(t)$  (and hence $\phi^*(t)$) 
that yields (\ref{eq:drift-noniid}), (\ref{eq:near-profit}) gives: 
\begin{eqnarray}
&\Delta_T(H(t_0)) - V\sum_{\tau=t_0}^{t_0+T-1} \expect{\phi(\tau)\left|\right.H(t_0)} \leq DT^2\nonumber \\
& -VT\phi^{opt} + VT\epsilon \nonumber \\
& +\sum_{n=1}^N[Vp_n^{max} + \mu_n^{max}]T\epsilon \label{eq:the-above} 
\end{eqnarray}
where we have used the fact that (by (\ref{eq:q-bound}) and (\ref{eq:theta-n})): 
\[ |Q_n(t_0) - \theta_n| \leq Vp_n^{max} + \mu_n^{max}    \]
Taking expectations of (\ref{eq:the-above})  with respect to $H(t_0)$ yields: 
\begin{eqnarray*}
\expect{L(\bv{Q}(t_0+T)) - L(\bv{Q}(t_0))} - V\sum_{\tau=t_0}^{t_0+T-1} \expect{\phi(\tau)} \leq \\
C_1T^2 + VTC_2 \epsilon - VT\phi^{opt} 
\end{eqnarray*}
where $C_1$ and $C_2$ are defined: 
\begin{eqnarray*}
C_1 &\defequiv& D  + \frac{\epsilon}{T} \sum_{n=1}^N\mu_n^{max} \\
C_2 &\defequiv& 1 + \sum_{n=1}^N p_n^{max} 
\end{eqnarray*}
The above holds for all $t_0$. Summing over $t_0 \in \{0, T, 2T, \ldots, (M-1)T\}$ for some positive integer $M$
and dividing by $VMT$ yields: 
\begin{eqnarray*}
\frac{\expect{L(\bv{Q}(MT)) - L(\bv{Q}(0))}}{VMT} - \frac{1}{MT}\sum_{\tau=0}^{MT-1} \expect{\phi(\tau)} \leq \\
C_1T/V +  C_2\epsilon - \phi^{opt} 
\end{eqnarray*}
Rearranging terms and using non-negativity of $L(\cdot)$ yields: 
\begin{eqnarray*}
\overline{\phi}(MT) \geq \phi^{opt} - C_2\epsilon - C_1T/V - \frac{\expect{L(\bv{Q}(0))}}{VMT} 
\end{eqnarray*}
Therefore (noting that the $\liminf$ sampled every $T$ slots is the same as the regular
$\liminf$ because $\phi(\tau)$ is bounded)  yields: 
\[ \liminf_{t\rightarrow\infty} \overline{\phi}(t) \geq \phi^{opt} - C_2\epsilon - C_1T/V \]

  \section*{Appendix C  --- Characterization of $\phi^{opt}$} 
  
  \begin{lem} The value $\phi^{opt}$ is achievable by a single $p$-only policy that satisfies
  $d_n^*=0$ for all $n \in \{1, \ldots, N\}$.
  \end{lem}

\begin{proof}  
For each price vector $\bv{p}$ in the finite set $\script{P}$,  define $\Omega(\bv{p})$ as the set of 
all decision vectors $[\bv{A}; \bv{\mu}]$ that satisfy (\ref{eq:mu1}), (\ref{eq:mu2}), (\ref{eq:A1}), (\ref{eq:A2}), where $\bv{p}(t)$ is replaced
with $\bv{p}$ in (\ref{eq:mu2}) and (\ref{eq:A2}). Note that $\Omega(\bv{p})$ is finite for each $\bv{p} \in \script{P}$. 
A $p$-only policy is characterized by a conditional probability distribution $q(\bv{A}, \bv{\mu}|\bv{p})$
that satisfies: 
\begin{eqnarray} 
\sum_{[\bv{A}; \bv{\mu}] \in \Omega(\bv{p})} q(\bv{A}, \bv{\mu}|\bv{p}) = 1 & \mbox{ for all $\bv{p} \in \script{P}$} \label{eq:ponly1} \\
0 \leq q(\bv{A}, \bv{\mu}|\bv{p}) \leq 1 & \mbox{ for all $\bv{A}, \bv{\mu}, \bv{p}$} \label{eq:ponly2} \\
q(\bv{A}, \bv{\mu}|\bv{p})  = 0 & \mbox{ whenever $[\bv{A}; \bv{\mu}] \notin \Omega(\bv{p})$} \label{eq:ponly3} 
\end{eqnarray} 
where $q(\bv{A}, \bv{\mu}|\bv{p})$ is defined: 
\[ q(\bv{A}, \bv{\mu}|\bv{p}) \defequiv Pr[\bv{A}(t) = \bv{A}, \bv{\mu}(t) = \bv{\mu} \left|\right.\bv{p}(t) = \bv{p}] \]

The collection of 
values $q(\bv{A}, \bv{\mu}|\bv{p})$ for $\bv{p} \in \script{P}$ and $[\bv{A}; \bv{\mu}] \in \Omega(\bv{p})$ 
can be viewed as a finite dimensional vector defined over the compact 
set defined by (\ref{eq:ponly1})-(\ref{eq:ponly3}).   Hence, by the Bolzano-Wierstrass theorem, 
any infinite sequence of such policies must have a convergent subsequence that converges to a particular
$p$-only policy that satisfies (\ref{eq:ponly1})-(\ref{eq:ponly3}).  In particular, let 
$\bv{A}^{(k)}(t)$, $\bv{\mu}^{(k)}(t)$ be an infinite sequence of $p$-only policies defined by 
distributions $q^{(k)}(\bv{A}, \bv{\mu}|\bv{p})$ that satisfy (\ref{eq:ponly1})-(\ref{eq:ponly3}), and define: 
\begin{eqnarray*}
&&\hspace{-.3in}d_n^{(k)} \defequiv \sum_{\bv{p}\in\script{P}} \pi(\bv{p}) \sum_{[\bv{A}; \bv{\mu}] \in \Omega(\bv{p})} q^{(k)}(\bv{A}, \bv{\mu}|\bv{p})[A_n - \mu_n] \\
&&\hspace{-.3in}\phi^{(k)} \defequiv \sum_{\bv{p}\in\script{P}} \pi(\bv{p})  \sum_{[\bv{A}; \bv{\mu}] \in\Omega(\bv{p})} q^{(k)}(\bv{A}, \bv{\mu}|\bv{p})\sum_{n=1}^{N} [\mu_np_n - s_n(\mu_n)] \\
&&-  \sum_{\bv{p}\in\script{P}} \pi(\bv{p})  \sum_{[\bv{A}; \bv{\mu}] \in\Omega(\bv{p})}q^{(k)}(\bv{A}, \bv{\mu}|\bv{p})\sum_{n=1}^N[A_np_n + b_n(A_n)]
\end{eqnarray*}
It is clear that $d_n^{(k)}$ and $\phi^{(k)}$ correspond to the virtual drift of stock $n$ and the virtual profit under the 
$p$-only policy $\bv{A}^{(k)}(t)$, $\bv{\mu}^{(k)}(t)$, as defined by the time average expectations in (\ref{eq:dn-star}), (\ref{eq:phi-star}).
Assume that this infinite sequence of $p$-only policies satisfies: 
\begin{eqnarray}
&d_n^{(k)} \geq 0 \: \:  \mbox{ for all $n \in \{1, \ldots, N\}, k\in\{0, 1, \ldots\}$}& \label{eq:dn-k} \\
&\lim_{k\rightarrow\infty} \phi^{(k)} = \phi^{opt}&  \label{eq:phi-lim} 
\end{eqnarray}  
Consider now any convergent subsequence of distributions $q^{(k_m)}(\bv{A}, \bv{\mu}|\bv{p})$ that converge to 
some particular distribution $q^*(\bv{A}, \bv{\mu}|\bv{p})$ that satisfies (\ref{eq:ponly1})-(\ref{eq:ponly3}).  This defines
a single $p$-only policy.  Further, by (\ref{eq:dn-k})-(\ref{eq:phi-lim}), this $p$-only policy must satisfy: 
\[ d_n^* \geq 0 \: \: \mbox{ for all $n \in \{1, \ldots, N\}$} \: \: , \: \: \phi^* = \phi^{opt} \]

It remains only to show that the algorithm can be modified to achieve $\phi^{opt}$ with 
$d_n^* = 0$ for all $n \in \{1, \ldots, N\}$.  Suppose the current $p$-only policy 
has a stock $n \in \{1, \ldots, N\}$ such that $d_n^*>0$.  We shall create a new $p$-only policy with $d_n^*=0$, without reducing
profit.    Define: 
\begin{eqnarray*} 
 \alpha_n^*  &\defequiv& \sum_{\bv{p} \in \script{P}} \pi(\bv{p}) \sum_{[\bv{A}; \bv{\mu}]\in\Omega(\bv{p})} q^*(\bv{A}, \bv{\mu}|\bv{p})A_n \\
 \beta_n^* &\defequiv& \sum_{\bv{p} \in \script{P}} \pi(\bv{p}) \sum_{[\bv{A}; \bv{\mu}]\in\Omega(\bv{p})} q^*(\bv{A}, \bv{\mu}|\bv{p})\mu_n 
 \end{eqnarray*}
 Then $d_n^* = \alpha_n^* - \beta_n^*$, and so $\alpha_n^* > \beta_n^* \geq 0$.   Consider now a \emph{new} $p$-only policy 
 $\tilde{\bv{A}}(t)$, $\tilde{\bv{\mu}}(t)$
 defined as follows:  Define $\tilde{\bv{\mu}}(t) \defequiv \bv{\mu}^*(t)$ (so that selling decisions are the same).  Define
 $\tilde{A}_m(t) \defequiv A_m^*(t)$ 
 for all $m \neq n$.  For stock $n$, choose $\tilde{A}_n(t)$ as follows: 
 \[ \tilde{A}_n(t) \defequiv  \left\{ \begin{array}{ll}
                          A_n^*(t) &\mbox{ with probability $\beta_n^*/\alpha_n^*$} \\
                             0  & \mbox{ otherwise} 
                            \end{array}
                                 \right. \]
 Note that this new $p$-only policy satisfies the constraints (\ref{eq:mu1}), (\ref{eq:mu2}), (\ref{eq:A1}), (\ref{eq:A2}), as the original 
 policy satisfies these constraints, and we have only changed the $\bv{A}^*(t)$ decision vector
 by probabilistically setting the $n$th entry to zero. Also note that the drift for all stocks $m \neq n$ is unchanged, so that 
 $\tilde{d}_m \geq 0$ for all $m \neq n$.  Further: 
 \[ \tilde{d}_n = \alpha_n^*(\beta_n^*/\alpha_n^*) - \beta_n^* = 0 \]
 Thus, we have $\tilde{d}_m \geq 0$ for all $m \in \{1, \ldots, N\}$.  Finally, it is easy to see that this modification has not reduced  the 
 profit value, and hence it must also achieve $\tilde{\phi} = \phi^{opt}$.     If there are any remaining stocks $m$ such that $d_m^*>0$, 
 we can repeat the same modification procedure. This proves the existence of a $p$-only policy that achieves
 $\phi^{opt}$ with $d_n^* = 0$ for all $n \in \{1, \ldots, N\}$. 
  \end{proof} 
  
\begin{lem}  If the price process $\bv{p}(t)$ satisfies (\ref{eq:time-avg-p}), then 
$\phi^{opt}$ is an upper bound on the $\limsup$
  time average profit of any policy that satisfies (\ref{eq:mu1})-(\ref{eq:A2}).  In particular, if $\bv{A}(t)$ and $\bv{\mu}(t)$ are decisions
  for any policy that satisfies (\ref{eq:mu1})-(\ref{eq:A2}) for all $t \in \{0, 1, 2, \ldots\}$, then: 
  \begin{equation} \label{eq:nec1} 
   \limsup_{t\rightarrow\infty} \frac{1}{t} \sum_{\tau=0}^{t-1} \phi(\tau) \leq \phi^{opt} \: \: \mbox{ with probability 1}
   \end{equation} 
  and: 
  \begin{equation} \label{eq:nec2} 
  \limsup_{t\rightarrow\infty} \frac{1}{t}\sum_{\tau=0}^{t-1} \expect{\phi(\tau)} \leq \phi^{opt}
  \end{equation}
  \end{lem} 
  \begin{proof} 
  We prove only (\ref{eq:nec1}) (the result (\ref{eq:nec2}) follows from (\ref{eq:nec1}), for example, using the Lebesgue Dominated
  Convergence Theorem with the observation that $0 \leq \phi(\tau)  \leq  \sum_{n=1}^N p_n^{max} \mu_n^{max}$).
  Because the algorithm can never sell more stock than it has, for a given time $t$ we have: 
  \begin{equation} \label{eq:cannot-sell-more} 
   \sum_{\tau=0}^{t-1} A_n(\tau) \geq \sum_{\tau=0}^{t-1} \mu_n(\tau) \: \: \mbox{ for all $n \in \{1, \ldots, N\}$}
   \end{equation} 
   Now for each $\bv{p} \in \script{P}$, define $T_{\bv{p}}(t)$ as the set of slots $\tau \in \{0, 1, \ldots, t-1\}$ 
   for which $\bv{p}(\tau) = \bv{p}$, and define $|T_{\bv{p}}(t)|$ as the total number of such slots. Define $\script{P}(t)$ as the set of all price vectors $\bv{p} \in \script{P}$ for which 
   $|T_{\bv{p}}(t)| > 0$. 
   We thus have: 
   \begin{eqnarray*}
   \frac{1}{t}\sum_{\tau=0}^{t-1} \phi(\tau) = \sum_{\bv{p}\in\script{P}(t)} \frac{|T_{\bv{p}}(t)|}{t} \frac{1}{|T_{\bv{p}}(t)|}\sum_{\tau \in T_{\bv{p}}(t)} \phi(\tau) 
   \end{eqnarray*}
   However, for each $\bv{p} \in \script{P}(t)$ we have: 
   \begin{eqnarray*}
   \frac{1}{|T_{\bv{p}}(t)|}\sum_{\tau \in T_{\bv{p}}(t)} \phi(\tau) = \\
   \frac{1}{|T_{\bv{p}}(t)|} \sum_{[\bv{A}; \bv{\mu}] \in \Omega(\bv{p})} N(\bv{A}, \bv{\mu}, \bv{p}, t)\hat{\phi}(\bv{A}, \bv{\mu}, \bv{p}) 
   \end{eqnarray*}
   where $N(\bv{A}, \bv{\mu}, \bv{p}, t)$ is defined as the number of times during the interval $\{0, \ldots, t-1\}$ 
   that the algorithm 
   selects $\bv{A}(\tau) = \bv{A}$, $\bv{\mu}(\tau) = \bv{\mu}$ when $\bv{p}(\tau) = \bv{p}$, and where $\hat{\phi}(\bv{A}, \bv{\mu}, \bv{p})$ 
   is given by: 
   \begin{eqnarray*}
    \hat{\phi}(\bv{A}, \bv{\mu}, \bv{p}) \defequiv \sum_{n=1}^N [\mu_np_n - s_n(\mu_n)] 
    - \sum_{n=1}^N[ A_np_n + b_n(A_n)] 
    \end{eqnarray*}
   The values $N(\bv{A}, \bv{\mu}, \bv{p}, t)$
   define a $p$-only policy, given by distribution: 
   \[ q^{(t)}(\bv{A}, \bv{\mu}|\bv{p}) =\left\{ \begin{array}{ll}
                          \frac{N(\bv{A}, \bv{\mu}, \bv{p}, t)}{|T_{\bv{p}}(t)|}   &\mbox{ if $|T_{\bv{p}}(t)|>0$} \\
                             0  & \mbox{ otherwise} 
                            \end{array}
                                 \right. \]
   Further, this distribution satisfies the constraints (\ref{eq:ponly1})-(\ref{eq:ponly3}) required for $p$-only policies. 
   Now let $t_k$ be an infinite subsequence over which the $\limsup$ time average profit is achieved, so that: 
   \[ \limsup_{t\rightarrow\infty} \frac{1}{t}\sum_{\tau=0}^{t-1} \phi(\tau) = \lim_{k\rightarrow\infty} \frac{1}{t_k}\sum_{\tau=0}^{t_k} \phi(\tau) \]
   We thus have: 
   \begin{eqnarray}
  &&\hspace{-.4in} \frac{1}{t_k}\sum_{\tau=0}^{t_k} \phi(\tau) = \nonumber \\
  && \hspace{-.2in}\sum_{\bv{p}\in\script{P}(t_k)} \frac{|T_{\bv{p}}(t_k)|}{t_k} 
   \sum_{[\bv{A}; \bv{\mu}]\in\Omega(\bv{p})} q^{(t_k)}(\bv{A}, \bv{\mu}|\bv{p})\hat{\phi}(\bv{A}, \bv{\mu}, \bv{p}) \label{eq:c1} 
   \end{eqnarray}
   Further, with this notation, from (\ref{eq:cannot-sell-more}) we have for each $n \in \{1, \ldots, N\}$: 
   \begin{eqnarray}
&&\hspace{-.45in}0 \leq    \sum_{\tau=0}^{t_k-1} [A_n(\tau) - \mu_n(\tau)] \nonumber \\
&&\hspace{-.35in}=  \sum_{\bv{p}\in\script{P}(t_k)} \frac{|T_{\bv{p}}(t_k)|}{t_k} 
   \sum_{[\bv{A}; \bv{\mu}]\in\Omega(\bv{p})} q^{(t_k)}(\bv{A}, \bv{\mu}|\bv{p})[A_n - \mu_n] \label{eq:c2} 
   \end{eqnarray}
   Because $\script{P}$ is finite and $\Omega(\bv{p})$ is finite for each $\bv{p} \in \script{P}$, the $p$-only distributions
   $q^{(t_k)}(\bv{A}, \bv{\mu}|\bv{p})$ can be viewed as an infinite sequence of vectors in a compact set defined
   by (\ref{eq:ponly1})-(\ref{eq:ponly3}), and hence have
   a convergent subsequence that converges to a distribution $q^*(\bv{A}, \bv{\mu}|\bv{p})$ that is in the 
   set (\ref{eq:ponly1})-(\ref{eq:ponly3}). 
   Note by (\ref{eq:time-avg-p})  that  for each $\bv{p} \in \script{P}$ we have: 
  \[ \lim_{t\rightarrow\infty} \frac{|T_{\bv{p}}(t)|}{t} = \pi(\bv{p}) \: \: \mbox{ with probability 1} \]
   Taking limits of (\ref{eq:c1}) and (\ref{eq:c2}) thus yields: 
    \begin{eqnarray*}
 && \limsup_{t\rightarrow\infty} \frac{1}{t}\sum_{\tau=0}^{t-1}\phi(\tau) = \\
  &&\sum_{\bv{p}\in\script{P}} \pi(\bv{p}) 
   \sum_{[\bv{A}; \bv{\mu}]\in\Omega(\bv{p})} q^{*}(\bv{A}, \bv{\mu}|\bv{p})\hat{\phi}(\bv{A}, \bv{\mu}, \bv{p}) \label{eq:cc1} 
   \end{eqnarray*}
   and for all $n \in \{1,\ldots, N\}$: 
    \begin{eqnarray*}
    0  \leq     \sum_{\bv{p}\in\script{P}}\pi(\bv{p}) 
   \sum_{[\bv{A}; \bv{\mu}]\in\Omega(\bv{p})} q^{*}(\bv{A}, \bv{\mu}|\bv{p})[A_n - \mu_n] \defequiv d_n^*
   \end{eqnarray*}
   This defines a $p$-only policy that achieves the $\limsup$ time average of $\phi(t)$, while yielding 
   $d_n^*\geq 0$ for all $n$.   It follows that the $\limsup$ time average of $\phi(t)$ must be less than or equal to the
   value $\phi^{opt}$ defined as the largest such value achievable over $p$-only policies that satisfy $d_n^* \geq 0$ for all $n$.
  \end{proof}
\bibliographystyle{unsrt}
\bibliography{../../latex-mit/bibliography/refs}
\end{document}